%% file: main.tex
\documentclass[runningheads]{llncs}
\usepackage[margin=2.5cm]{geometry}

\usepackage{graphicx,wrapfig,color,colortbl,bm,cite,sidecap,xspace}
\usepackage[export]{adjustbox}
\usepackage{amsmath}
\usepackage{amssymb}
\usepackage{authblk}
\usepackage{braket}
\usepackage{color}
\usepackage{colortbl}
\usepackage{float}
\usepackage[labelfont=bf,small]{caption}
\usepackage{mathtools}
\usepackage{pgfplots}
\usepackage{physics}
\usepackage{stmaryrd}
\usepackage{subcaption}
\usepackage[normalem]{ulem}
\usepackage{scalerel}
\usepackage{hyperref}

\usepackage{tikzit}
\input{zx.tikzstyles}
\input{zx.tikzdefs}

\DeclareUnicodeCharacter{2212}{-}

\begin{document}

\title{Constructing all qutrit controlled Clifford+\texorpdfstring{$T$}{T} gates in Clifford+\texorpdfstring{$T$}{T}}

\author{Lia Yeh\inst{1}
\and
John van de Wetering\inst{1,2}
}
\authorrunning{L.~Yeh \& J.~van de Wetering}

\institute{University of Oxford \and
Radboud University Nijmegen}

\maketitle

\begin{abstract}
    For a number of useful quantum circuits, qudit constructions have been found which reduce resource requirements compared to the best known or best possible qubit construction.
    However, many of the necessary qutrit gates in these constructions have never been explicitly and efficiently constructed in a fault-tolerant manner.
    We show how to exactly and unitarily construct any qutrit multiple-controlled Clifford+$T$ unitary using just Clifford+$T$ gates and without using ancillae.
    The $T$-count to do so is polynomial in the number of controls $k$, scaling as $O(k^{3.585})$.
    With our results we can construct ancilla-free Clifford+$T$ implementations of multiple-controlled $T$ gates as well as all versions of the qutrit multiple-controlled Toffoli, while the analogous results for qubits are impossible.
    As an application of our results, we provide a procedure to implement any ternary classical reversible function on $n$ trits as an ancilla-free qutrit unitary using $O(3^n n^{3.585})$ $T$ gates.

    \keywords{Qutrits \and Gate Synthesis \and Clifford+T}
\end{abstract}

\section{Introduction}
\label{sec:intro}

Classical computing technology works with bits, where the state of the fundamental information unit can be in one of two states. It is then not surprising that quantum computing researchers have mostly studied \emph{qu}bits, where the fundamental unit of information can be in a superposition of two states.
However, there are several benefits we can get by working instead with qu\emph{dits}, where we work with higher-dimensional systems.
One such benefit is that many proposed physical types of qubits are actually restricted subspaces of higher-dimensional systems, where the natural dimension can be much higher. By working with qudits we can exploit the additional degrees of freedom present in the system.  
Qudit quantum processors based on ion traps~\cite{RingbauerM2021quditions} and superconducting devices~\cite{BlokM2021scrambling,YeB2018cphasephoton,YurtalanM2020Walsh-Hadamard} have already been demonstrated.
By using the otherwise wasted dimensions accessible in a qudit we increases the device's information density, which leads to advantages in for instance runtime efficiency, resource requirements, magic state distillation noise thresholds, and noise resilience in communication~\cite{WangY2020quditsreview,CozzolinoD2019quditcommunication,CampbellE2014quditmsdthresholds,PrakashS2020qutritmsd}.

For qudits to make a good foundation for a quantum computer, we need techniques to do fault-tolerant computation with them. 
A well-studied approach for fault-tolerant computation with qubits is based on the observation that many quantum error correcting codes can natively implement Clifford gates, so that we only need to realise a fault-tolerant implementation of some non-Clifford gate. A popular choice for this gate is the $T$ gate, a single-qubit gate that can be implemented by injecting its magic state into a circuit. As these magic states can be distilled to a desired level of fidelity, we can then implement approximately universal quantum computation fault-tolerantly~\cite{BravyiS2005ftqc}. 
Qudit analogues of this Clifford+$T$ gate set have been developed, so that this approach of magic state distillation and injection can be used to do fault-tolerant computation for qudits of any dimension~\cite{CampbellE2012tgatedistillation}.

In this paper we will work with qu\emph{trits}, three-dimensional systems. These are the most well-studied higher-dimensional qudit.
Qutrits have been used to reduce the circuit complexity of implementing multi-controlled qubit gates~\cite{GokhaleP2019asymptotic, KiktenkoE2020quditroutinggraph, IonicioiuR2009mctqudit, RalphT2007qudittoffolioptics, LanyonB2009multilevelinformationcarrier}.
By replacing some or all of the information carriers to be qutrits instead of qubits, the $\ket{2}$ energy level of the qutrit can be utilized to reduce resource requirements in terms of number of ancillae, entanglement complexity, gate depth and gate count, and non-Clifford gate count.
For instance, in Ref.~\cite{GokhaleP2019asymptotic} they showed how to implement an $n$-controlled (qubit) Toffoli gate in $O(\log n)$ depth, using just $O(n)$ two-qutrit gates and no ancillae. More generally, we say a qutrit circuit \emph{emulates} a qubit gate when the action on the $\{\ket{0}, \ket{1}\}$ subspace is  equal to the qubit gate. Emulating a qubit logic gate with a qutrit unitary can be more efficient than only using qubits.  This is because we can utilise the additional $\ket{2}$ state as an intermediate storage to decrease the cost of implementation.
The gates used in these constructions involve what we will call \emph{$\ket{2}$-controlled gates} --- controlled unitaries that only fire when the control qutrit is in the $\ket{2}$ state, so that this gate acts as the identity on the target if the control is in the $\ket{0}$ or $\ket{1}$ state.
In fact, many qutrit-native algorithms use $\ket{2\cdots 2}$-controlled logic gates, including various ternary adders and incrementers~\cite{KhanM2007ternaryadder,HaghparastM2017ternaryadder,GokhaleP2019asymptotic}.

Because all these results are based on qutrit-controlled gates, it is crucial to understand how we can implement these fault-tolerantly if we wish to use them for practical purposes.
Unfortunately, while there is a `naive' decomposition into qutrit Clifford+$T$ gates of for instance the $\ket{2\cdots 2}$-controlled $X$ gate which uses $O(n)$ clean ancilla for $n$ ternary controls~\cite{KhanM2007ternaryadder}, ancilla-free implementations require either uniformly controlled Givens rotations~\cite{KhanF2006synthesisqudit} or qutrit-controlled qubit gates~\cite{DiY2013synthesis}, both of which are not fault-tolerant (at least as stated).  The construction of Ref.~\cite{MoragaC2016quditgeneralizedtoffoli} is conceivably fault-tolerant, but it utilizes an exponential number of gates.
This raises the question of how we can implement these $\ket{2\cdots 2}$-controlled gates efficiently using more primitive and fault-tolerant gates.

In this paper we show that when we have any qutrit Clifford+$T$ unitary $U$, we can construct an ancilla-free exact Clifford+$T$ implementation of the $\ket{2\cdots 2}$-controlled $U$ unitary which uses a number of gates polynomial in the number of controls. Specifically, for $k$ controls we require $O(k^{3.585})$ gates (this number comes from $\log_2 6 \approx 3.585$). 
Our work means in particular that we have fault-tolerant and ancilla-free implementations of all the constructions mentioned above.
Note that our result, constructing controlled Clifford+$T$ unitaries for any Clifford+$T$ unitary, is not possible with qubits when we don't allow ancillae. For instance, it is not possible to construct a qubit controlled-$T$ gate~\cite{KliuchnikovV2013singlequbitcliffordplust} or a three-controlled Toffoli~\cite{GilesB2013multiqubitcliffordplustsynthesis} using just Clifford+$T$ gates and no ancillae.
The constructions in this paper build on our work in a previous paper where we showed how to construct the qutrit single-controlled Hadamard and $S$ gates, which we used to exactly synthesise the qutrit metaplectic gate~\cite{GlaudellA2022qutritmetaplecticsubset}.
A software implementation of some of our constructions can be found on Github%
\footnote{\url{https://github.com/lia-approves/qudit-circuits/tree/main/qutrit_control_Clifford_T}}.

As an application of our construction we give an algorithm for implementing any reversible classical trit function $f\colon\{0,1,2\}^n\to \{0,1,2\}^n$ as a unitary $n$-qutrit Clifford+$T$ circuit using at most $O(3^n n^{3.585})$ Clifford+$T$ gates. We find a lower bound for this problem of $O(3^n \cdot n/\log n)$ so that our result here is within a polynomial factor of optimal.

The paper is structured as follows.
In Section~\ref{sec:preliminaries}, we recall the basics of the qutrit Clifford+$T$ gate set and the different types of control wires for qutrit unitaries, and we recall several known results for controlled qutrit unitaries.
Then in Section~\ref{sec:controlcliffordplust}, we present exact ancilla-free Clifford+$T$ constructions of any $\ket{2}^{\otimes n}$-controlled Clifford+$T$ unitary.
In Section~\ref{sec:permutations} we show how we can use our results to implement any ternary classical reversible function as a Clifford+$T$ ancilla-free unitary.
Finally, we end with some concluding remarks in Section~\ref{sec:conclusion}.

\section{Preliminaries}\label{sec:preliminaries}

A qubit is a two-dimensional Hilbert space. Similarly, a qutrit is a three-dimensional Hilbert space. We will write $\ket{0}$, $\ket{1}$, and $\ket{2}$ for the standard computational basis states of a qutrit.
Any normalised qutrit state can then be written as $\ket{\psi} = \alpha \ket{0} +  \beta \ket{1} + \gamma \ket{2}$ where $\alpha,\beta,\gamma\in \mathbb{C}$ and $|\alpha|^2 + |\beta|^2 + |\gamma|^2 = 1$.

Several concepts for qubits extend to qutrits, or more generally to qu\emph{dits}, which are $d$-dimensional quantum systems. In particular, the concept of Pauli's and Cliffords.
\begin{definition}
For a $d$-dimensional qudit, the Pauli $X$ and $Z$ gates are defined as
\begin{equation}
    X\ket{k} = \ket{k+1} \qquad\qquad Z\ket{k} = \omega^k \ket{k}
\end{equation}
where $\omega\coloneqq e^{2\pi i/d}$ is such that $\omega^d = 1$, and the addition $\ket{k+1}$ is taken modulo $d$~\cite{GottesmanD1999ftqudit,HowardM2012quditTgate}.
We define the \emph{Pauli group} as the set of unitaries generated by tensor products of the $X$ and $Z$ gate. 
We write $\mathcal{P}_n^d$ for the Pauli's acting on $n$ qudits.
\end{definition}
For qubits this $X$ gate is just the NOT gate, while $Z=\text{diag}(1,-1)$. For the duration of this paper we will work solely with qutrits, so we take $\omega$ to always be equal to $e^{2\pi i/3}$.

For a qubit there is only one non-trivial permutation of the standard basis states, which is implemented by the $X$ gate.
For qutrits there are five non-trivial permutations of the basis states. By analogy we will all call these ternary $X$ gates. These gates are $X_{+1}$, $X_{-1}$, $X_{01}$, $X_{12}$, and $X_{02}$. 
The gate $X_{\pm 1}$ sends $\ket{t}$ to $\ket{(t \pm 1) \text{ mod } 3}$ for $t \in \{0, 1, 2\}$; $X_{01}$ is just the qubit $X$ gate which is the identity when the input is $\ket{2}$; $X_{12}$ sends $\ket{1}$ to $\ket{2}$ and $\ket{2}$ to $\ket{1}$, and likewise for $X_{02}$.
Note that the qutrit Pauli $X$ gate is the $X_{+1}$ gate, while $X^\dagger = X^2 = X_{-1}$.

\subsection{The Clifford+\texorpdfstring{$T$}{T} gate set}

Another concept that translates to qutrits (or more general qudits) is that of Clifford unitaries.

\begin{definition}
    Let $U$ be a qudit unitary acting on $n$ qudits. We say it is \emph{Clifford} when every Pauli is mapped to another Pauli under conjugation by $U$. I.e.~if for any $P\in \mathcal{P}_n^d$ we have $UPU^\dagger \in \mathcal{P}_n^d$.
\end{definition}

Note that the set of $n$-qudit Cliffords forms a group under composition. For qubits, this group is generated by the $S$, Hadamard and CX gates. The same is true for qutrits, for the right generalisation of these gates.

Throughout the paper we will write $\zeta$ for the ninth root of unity $\zeta = e^{2\pi i/9}$. Note that $\zeta^3 = \omega$ and $\zeta^9=1$.

\begin{definition}
    The qutrit $S$ gate is $S\coloneqq \zeta^8 \text{diag}(1,1,\omega)$. I.e.~it multiplies the $\ket{2}$ state by the phase $\omega$ (up to a global phase).
\end{definition}
We adopt the convention of this global phase of $\zeta^8$ from Ref.~\cite{GlaudellA2019canonical}, as it will make some of our results more elegant to state (without it we would often have to say `up to global phase').

For qubits, the Hadamard gate interchanges the $Z$ basis, $\ket{0}$ and $\ket{1}$ which are the eigenstates of the Pauli $Z$, and the $X$ basis, consisting of the states $\ket{\pm} \coloneqq \frac{1}{\sqrt{2}}(\ket{0}\pm \ket{1})$. The same holds for the qutrit Hadamard.
In this case the $X$ basis consists of the following states:
\begin{align}
    \ket{+} &\coloneqq \frac{-i}{\sqrt{3}} (\ket{0}+\ket{1}+\ket{2}) \\
    \ket{\omega} &\coloneqq \frac{-i}{\sqrt{3}} (\ket{0}+\omega\ket{1}+\omega^2\ket{2}) \\
    \ket{\omega^2} &\coloneqq \frac{-i}{\sqrt{3}} (\ket{0}+\omega^2\ket{1}+\omega\ket{2})
\end{align}

\begin{definition}
    The \emph{qutrit Hadamard gate} $H$ is the gate that maps $\ket{0} \mapsto \ket{+}$, $\ket{1}\mapsto \ket{\omega}$ and $\ket{2} \mapsto \ket{\omega^2}$. As a matrix:
    \begin{equation}\label{eq:hgatedef}
        H \ \coloneqq \ \frac{-i}{\sqrt{3}}\begin{pmatrix}
            1 & 1 & 1 \\
            1 & \omega & \omega^2 \\
            1 & \omega^2 & \omega
        \end{pmatrix}
    \end{equation}
\end{definition}
We chose the global phase of the $H$ gate to be $-i$ to be in line with Refs.~\cite{GlaudellA2019canonical,GlaudellA2022qutritmetaplecticsubset}.

Note that, unlike the qubit Hadamard, the qutrit Hadamard is \emph{not} self-inverse. 
Instead we have $H^2 = -X_{12}$ so that $H^4 = \mathbb{I}$. This means that $H^\dagger = H^3$.

We have $Z = H^\dagger X_{+1} H$, and hence we can call $Z=Z_{+1}$ by analogy.
We can then also define the `$Z$ permutation gates' by analogy. For instance, $Z_{01} \coloneqq H X_{01} H^\dagger$. It will in fact be helpful to define a larger class of $Z$ phase gates.
\begin{definition}\label{def:Z-phase-gate}
    We write $Z(a,b)$ for the \emph{phase gate} that acts as $Z(a,b)\ket{0} = \ket{0}$, $Z(a,b)\ket{1} = \omega^a\ket{1}$ and $Z(a,b)\ket{2} = \omega^b\ket{2}$ where we take $a,b\in \mathbb{R}$.
\end{definition}
We define $Z(a,b)$ in this way, taking $a$ and $b$ to correspond to phases that are multiples of $\omega$, because $Z(a,b)$ will turn out to be Clifford iff $a$ and $b$ are integers, so that we can easily see from the parameters whether the gate is Clifford or not. Note that the collection of all $Z(a,b)$ operators constitutes the group of diagonal single-qutrit unitaries modded out by a global phase. Composition of these operations is given by $Z(a,b)\cdot Z(c,d)=Z(a+b,c+d)$. Note that up to a global phase we have $S = Z(0,1)$.

In Definition~\ref{def:Z-phase-gate} we defined the Z phase gate. Similarly, we can define the X phase gates, that give a phase to the X basis gates.
\begin{definition}
    We define the \emph{X phase gates} to be $X(a,b) \coloneqq HZ(a,b)H^\dagger$ where $a,b \in \mathbb{R}$.
\end{definition}
We have in fact already seen examples of such X phase gates: $X_{+1} = X(2,1)$ and $X_{-1} = X(1,2)$.

Note that any single-qutrit Clifford can be represented (up to global phase) as a composition of Clifford $Z$ and $X$ phase gates.
In particular, we can represent the qutrit Hadamard in the following ways~\cite{GongX2017equivalence}:
\begin{align}
    \label{eq:EUrule}
    H &= Z(2,2)X(2,2)Z(2,2) = X(2,2)Z(2,2)X(2,2) \\
    H^\dagger &= Z(1,1)X(1,1)Z(1,1) = X(1,1)Z(1,1)X(1,1)
\end{align}
In analogy to its qubit counterpart, we will call these \emph{Euler decompositions} of the Hadamard.

The final Clifford gate we need is the qutrit CX.

\begin{definition}
    The qutrit CX is the two-qutrit gate defined by $\text{CX}\ket{i,j} = \ket{i,i+j}$ where the addition is taken modulo $3$.
\end{definition}

\begin{proposition}
    Let $U$ be a qutrit Clifford unitary. Then $U$ can be written as a composition of the $S$, $H$ and CX gates up to global phase.
\end{proposition}

Clifford gates are efficiently classically simulable, so we need to add another gate to get a universal gate set for quantum computing~\cite{GottesmanD1999ftqudit}.
This brings us to the definition of the qutrit $T$ gate.
\begin{definition}
    The qutrit $T$ gate is the $Z$ phase gate defined as $T \coloneqq \text{diag}(1,\zeta,\zeta^8)$~\cite{PrakashS2018normalform,CampbellE2012tgatedistillation,HowardM2012quditTgate}.
\end{definition}
Note that we could have written $T = Z(1/3,-1/3)$ as well.

Like the qubit $T$ gate, the qutrit $T$ gate belongs to the third level of the Clifford hierarchy, can be injected into a circuit using magic states, and its magic states can be distilled by magic state distillation. This means that we can fault-tolerantly implement this qutrit $T$ gate on many types of quantum error correcting codes.
Also as for qubits, the qutrit Clifford+$T$ gate set is approximately universal, meaning that we can approximate any qutrit unitary using just Clifford gates and the $T$ gate~\cite[Theorem 1]{CuiS2015universalmetaplectic}.

\subsection{Controlled unitaries}

When we have an $n$-qubit unitary $U$, we can speak of the controlled gate that implements $U$. This is the $(n+1)$-qubit gate that acts as the identity when the first qubit is in the $\ket{0}$ state, and implements $U$ on the last $n$ qubits if the first qubit is in the $\ket{1}$ state.

For qutrits there are multiple notions of control.

\begin{definition}\label{def:ket2-controlled}
Let $U$ be a qutrit unitary. Then the \emph{$\ket{2}$-controlled $U$} is the unitary $\ket{2}$-$U$ that acts as
\begin{equation*}
    \ket{0}\otimes \ket{\psi} \mapsto \ket{0}\otimes \ket{\psi} \qquad
    \ket{1}\otimes \ket{\psi} \mapsto \ket{1}\otimes \ket{\psi} \qquad
    \ket{2}\otimes \ket{\psi} \mapsto \ket{2}\otimes U\ket{\psi}
\end{equation*}
I.e.~it implements $U$ on the last qutrits if and only if the first qutrit is in the $\ket{2}$ state.
\end{definition}
Note that by conjugating the first qutrit with $X_{+1}$ or $X_{-1}$ gates we can make the gate also be controlled on the $\ket{1}$ or $\ket{0}$ state.

A different notion of qutrit control was introduced in Ref.~\cite{BocharovA2017ternaryshor}:
\begin{definition}\label{def:controlled-gates}
    Given a qutrit unitary $U$ we define
\begin{equation}
     \Lambda(U)\ket{c}\ket{t} = \ket{c} \otimes (U^c \ket{t}).
\end{equation}
I.e.~we apply the unitary $U$ a number of times equal to to the value of the control qutrit, so that if the control qutrit is $\ket{2}$ we apply $U^2$ to the target qutrits.
\end{definition}

The Clifford CX gate defined earlier is in this notation equal to $\Lambda(X_{+1})$.
Note that we can get this latter notion of control from the former one: just apply a $\ket{1}$-controlled $U$, followed by a $\ket{2}$-controlled $U^2$. 
Adding controls to a Clifford gate generally makes it non-Clifford. In the case of the CX gate, which is $\Lambda(X_{+1})$, it is still Clifford, while the $\ket{2}$-controlled $X_{+1}$ gate is \emph{not} Clifford.

A number of Clifford+$T$ constructions for controlled qutrit unitaries are already known. For instance, all the $\ket{2}$-controlled permutation $X$ gates can be built from the constructions given in Ref.~\cite{BocharovA2016ternaryarithmetics}, which we present here consistent with our notation.
The $\ket{0}$-controlled $Z$ gate can be constructed by the following 3 $T$ gate circuit, decomposition by Bocharov, Roetteler, and Svore~\cite[Figure 6]{BocharovA2017ternaryshor}:
\begin{equation}
    \label{eq:zczp1}
    \input{zczp1.tikz}
\end{equation}
Here the circles with a $1$, $2$, or $\Lambda$ inside denote controls of the types defined above.
By conjugating the control qutrit by either $X^\dagger$ before and $X$ after, or $X$ before and $X^\dagger$ after, the $\ket{1}$- and $\ket{2}$-controlled versions are respectively obtained.
Taking the adjoint of Eq.~\eqref{eq:zczp1} has the effect of changing the target operation from $Z$ to $Z^\dagger$.
Finally, the target can be changed to $X$ or to $X^\dagger$, by conjugating by a $H$ and $H^\dagger$ pair.
Using these gates we can build the $\ket{2}$-controlled $X_{01}$ gate:
\begin{equation}
        \label{eq:tcd01}
        \input{tcd01.tikz}
\end{equation}
Conjugating by the appropriate single-qutrit Clifford gates, these two circuits~\eqref{eq:zczp1} and~\eqref{eq:tcd01} suffice to construct any singly-controlled permutation $X$ or $Z$ gate.

The work done in Ref.~\cite{BocharovA2016ternaryarithmetics} also describes an approach which could be applied to constructing the $\ket{2}^{\otimes k}$-controlled $Z_{+1}$ gate for any $k$, and hence also the $\ket{2}^{\otimes k}$-controlled $Z_{-1}$, $X_{+1}$ and $X_{-1}$ gates for any $k$, by solving a system of linear equations modulo 3 where the number of equations is exponential in $k$.  They present the explicit circuit for this for $k = 1$, but not for any $k > 1$.  Their method does not suffice to construct the $k$-controlled $X_{01}$ gate due to the $X_{01}$ gate not being diagonal when conjugated by Hadamards.

Note that we usually don't care about global phases in the definition of unitaries. However, when adding controls, the global phase becomes `local' and hence is relevant~\cite[Lemma 5.2]{BarencoA1995elementarygates}.
\begin{definition}\label{def:controlledglobalphasegate}
     A \emph{controlled global phase gate} is a controlled unitary where the unitary is $e^{i \phi} \mathbb{I}$, for an identity matrix $I$ and some phase $\phi$.
 \end{definition}
The number of qutrits the identity matrix acts on in this definition is irrelevant as the phase factor can be ``factored out'' from the tensor product of the controlled and target qutrits:
 \begin{equation}
     \label{eq:controlledglobalphase}
     \input{controlledglobalphase.tikz}
 \end{equation}
Here we wrote the global phase $\phi$ as  $\phi = \gamma \cdot 2\pi/3$ so that we can represent the phase gate as a multiple of $\omega$.

We will see that it can be easier or more cost effective to construct a controlled unitary `up to a controlled global phase', and that to implement the unitary exactly, additional work must be done.  A generalisation of this idea was used for qubit controlled gates in Ref.~\cite{MaslovD2016relativephasetoffoli} to find more efficient decompositions of controlled gates.

In previous work~\cite{GlaudellA2022qutritmetaplecticsubset} we also found a construction of the $\ket{2}$-controlled $S$ gate and the $\ket{2}$-controlled Hadamard gate:
 \begin{equation}\label{eq:controlled-S-circ}
        \input{tcspdagphase.tikz}
\end{equation}
\begin{equation}\label{eq:controlled-had-circ}
    \input{controlled-had-circ.tikz}
\end{equation}

Note that the circuit for the controlled-$S$ gate that Eq.~\eqref{eq:controlled-S-circ} is based on in Ref.~\cite{GlaudellA2022qutritmetaplecticsubset} is only correct up to a controlled global phase. In this paper we defined our $S$ gate to include this global phase, so that it doesn't appear here. Due to this redefinition, the circuit Eq.~\eqref{eq:controlled-had-circ} acquires a (non-controlled) global phase of $\zeta$. This is a problem if we wish to use the gates in this circuit as a base for adding more controls, as this global phase then becomes a \emph{controlled} global phase. We will see later that there are however ways around this.

\section{Adding controls to Clifford+$T$ gates}
\label{sec:controlcliffordplust}
In this section we will implement each the $\ket{2}^{\otimes k}$-controlled versions of every qutrit Clifford+$T$ unitary.
As Clifford+$T$ unitaries are built out of Hadamard, CX, $S$ and $T$ gates, it suffices to show how we can construct $k$-controlled versions of each of these gates.

We will do this in stages, first showing how to construct the $X$ permutation gates with two controls, and then any number of controls, before moving on to the other Clifford+$T$ gates.

\subsection{Permutation gates with two controls}

Before we can make the step to having an arbitrary number of controls, we first need to construct the permutation gates with two controls.
First, we construct the $\ket{22}$-controlled $X_{01}$ gate. To do this we use the following lemma that allows us to add more controls to a unitary, once we know how to construct it with just one control.
 
\begin{lemma}
    \label{lemma:vvsqtrick}
    For any qutrit unitary $V$ with a construction for $\ket{2}$-controlled $V$ and $\ket{2}$-controlled $V^\dagger$, we can build the circuit consisting of the $\ket{22}$-controlled $V$ multiplied by  the $\ket{21}$-controlled $V^2$ unitary:
     \begin{equation}
     \textnormal{\input{c2c1vsqc2c2v.tikz}}
     \label{eq:c2c1vsqc2c2v}
    \end{equation}
\end{lemma}
\begin{corollary}
    \label{cor:vvsqtrick}
    If we pick $V=X_{01}$, then we have $V^2=\mathbb{I}$, so that this construction gives us a way to construct the $\ket{22}$-controlled $X_{01}$ using just $\ket{2}$-controlled $X_{01}$ and $X_{+1}$ and $X_{-1}$ gates, which we already know how to construct in Clifford+$T$.
\end{corollary}
Lemma~\ref{lemma:vvsqtrick} is easily shown to be correct by doing case distinctions on the control wires.
While we believe this construction to be new, it is based on the qubit Sleator-Weinfurter decomposition~\cite[Lemma 6.1]{BarencoA1995elementarygates}:
\begin{equation}
    \input{sleatorweinfurter.tikz}
\end{equation}
We note that Corollary~\ref{cor:vvsqtrick} contradicts the statement of Ref.~\cite{MoragaC2014basicternary} that ``a 5-gates Barenco et al. type of realization of a ternary (generalized) Toffoli gate without adding an ancillary line is simply not possible''; their analysis of ternary generalizations of the Sleator-Weinfurter decomposition did not account for decompositions of the form of Eq.~\eqref{eq:c2c1vsqc2c2v}.

\begin{lemma}
    We can construct the $\ket{22}$-controlled $X_{01}$ gate in Clifford+$T$ without using any ancilla.
\end{lemma}

Additionally, by conjugating by the appropriate single-qutrit Cliffords on the bottom qutrit, we can also construct the $\ket{22}$-controlled $X_{12,}$ and $X_{02}$ gates.

\begin{remark}
    The $\ket{22}$-controlled $X_{01}$ gate was decomposed by Di and Wei~\cite{DiY2013synthesis} in terms of qubit Clifford+$T$ operations, where the CX gate could be performed pairwise on any two of the three qutrit Z-basis states.
    As far as we are aware, there does not exist an error correction protocol that can, without code switching, correct qubit Clifford operations on all three pairs of Z-basis states. 
    The construction above is the first qutrit Clifford+$T$, and therefore suitable for fault-tolerant computation, decomposition of the $\ket{22}$-controlled $X_{01}$ gate.
\end{remark}

We can use this gate in turn to build the $\ket{22}$-controlled $X_{+1}$ gate.
To do this we will adapt a well-known construction for qubit controlled phases. Recall that if we can implement $k$-controlled Toffoli gates and the square root of a phase gate, that we can then also implement the $k$-controlled phase gate~\cite{BarencoA1995elementarygates}:
\begin{equation}\label{qubitgenmczphase}
    \input{qubitgenmczphase.tikz}
\end{equation}
This works because when the Toffoli `fires' the $X$ gate is applied, and we have $XR_Z(\alpha) X \propto Z(-\alpha)$.
\begin{lemma}
    \label{lemma:two-controlled-Xplus1}
    We can construct the $\ket{22}$-controlled $X_{+1}$ gate in Clifford+$T$ without using any ancillae.
\end{lemma}
\begin{proof}
    Consider the following circuit:
\begin{equation}\label{eq:two-controlled-Xplus1}
    \input{two-controlled-Xplus1.tikz}
\end{equation}
This works because $X_{12}X_{+1}X_{12} = X_{-1}$ and $X_{-1}^2 = X_{+1}$. 
\end{proof}
\begin{corollary}\label{cor:conj22}
    By conjugating by appropriate Clifford gates, we can use the $\ket{22}$-controlled $X_{+1}$ and $X_{01}$ gates to construct any $\ket{xy}$-controlled permutation gate where $x, y \in \{0,1,2\}$ of our choosing.
\end{corollary}

\subsection{Permutation gates with any number of controls}

Now let's see how we can generalise these constructions to have any number of controls. Here instead of building the controlled $X_{+1}$ gate out of the controlled $X_{01}$, we will go in the opposite direction.
In order to efficiently build the $\ket{2}^{\otimes k}$-controlled $X_{+1}$ gate, we will adapt a construction for multiple-controlled Toffolis for qubits that requires one \emph{borrowed ancilla}.
\begin{definition}
    A borrowed ancilla is an ancilla that can be in any state, and that is returned to the same state after the operation is finished.
\end{definition}
Having a borrowed ancilla just means that in the circuit we are considering there is at least one other qutrit that is not directly involved with our construction.

We base our construction on the following qubit identity for Toffolis~\cite[Lemma 7.3]{BarencoA1995elementarygates}:
\begin{equation}
    \input{toffoli-decomp.tikz}
\end{equation}
Using this construction, a $k$-controlled Toffoli where $k=2^m$ can be decomposed into 4 Toffoli's with $k/2=2^{m-1}$ controls. Iterating this procedure then requires $O(k^2)$ `standard' Toffolis with $2$ controls each.

\begin{lemma}\label{lem:controlled-Xplus1}
    We can construct the $\ket{2}^{\otimes k}$-controlled $X_{+1}$ gate using $O(k^{2.585})$ Clifford+$T$ gates and using one borrowed ancilla.
\end{lemma}
\begin{proof}
    Consider the following identity:
    \begin{equation}
    \input{Xplus1-decomp.tikz}
    \end{equation}
    To see this is correct first note that it is the identity when any of the top control wires is not in the $\ket{2}$ state (as doing $X_{+1}$ three times is just the identity).  As such, let us assume they are all in the $\ket{2}$ state so that we can ignore these control wires. We are cycling the value of the ancilla three times, which means that the controlled $X_{+1}$ gate on the target qutrit fires exactly once, when the ancilla is put into the $\ket{2}$ state. As we cycle it three times, the ancilla is put into the same state that it started in.

    We can use this to reduce the implementation of a $\ket{2}^{\otimes k}$-controlled $X_{+1}$ gate where $k=2^m$ to a sequence of 6 $\ket{2}^{\otimes k/2}$-controlled $X_{+1}$ gates. Letting $C(m)$ denote the number of $\ket{22}$-controlled $X_{+1}$ gates needed to write the $2^m$-controlled $X_{+1}$ gate, we then get the relation $C(m) = 6 C(m-1)$. As $C(1) = 1$ we calculate $C(m) = 6^{m-1} = \frac16 2^{m \log_2 6}$. Substituting $k=2^m$ we then see we require $O(k^{\log_2 6})$ of the base gate, which we can round up to $O(k^{2.585})$.
\end{proof}

\begin{lemma}\label{lem:many-controlled-X01}
    We can construct the $\ket{2}^{\otimes k}$-controlled $X_{01}$ gate using $O(k^{3.585})$ Clifford+$T$ gates without using any ancillae.
\end{lemma}
\begin{proof}
    Consider the following generalised version of Eq.~\eqref{eq:c2c1vsqc2c2v}, for which the qubit analogue is~\cite[Lemma 7.5]{BarencoA1995elementarygates}:
    \begin{equation}\label{eq:V-V2-many}
        \input{V-V2-many.tikz}
    \end{equation}
    Again taking $V=X_{01}$, we see that we can reduce the construction of the $\ket{2}^{\otimes k}$-controlled $X_{01}$ gate to the construction of the $\ket{2}^{\otimes k-1}$-controlled $X_{01}$ at the cost of introducing two $\ket{2}^{\otimes k-1}$-controlled $X_{+1}$ gates and two $\ket{2}$-controlled $X_{01}$ gates. Iterating the procedure we see then that the full construction requires $2k$ $\ket{2}$-controlled $X_{01}$ gates and $2k$ $\ket{2}^{\otimes m}$-controlled $X_{+1}$ gates where $m\leq k-1$. The cost of this is dominated by the controlled $X_{+1}$ gates, which cost $O(k^{2.585})$ gates each. As we have $O(k)$ of these, we require $O(k^{3.585})$ Clifford+$T$ gates to build the $\ket{2}^{\otimes k}$-controlled $X_{01}$ gate. Note that as the controlled $X_{+1}$ gates are not controlled on all the wires that we have access to at least one borrowed ancilla, so that we can in fact use the construction of Lemma~\ref{lem:controlled-Xplus1}.
\end{proof}

This construction doesn't require any borrowed ancillae. We can now use a generalised version of Eq.~\eqref{eq:two-controlled-Xplus1} to complete the circle and construct a many-controlled version of the $X_{+1}$ gate that does not require any borrowed ancillae (at the cost of worse polynomial scaling).

\begin{lemma}\label{lemma:gentwoxp1}
    We can construct the $\ket{2}^{\otimes k}$-controlled $X_{+1}$ gate using $O(k^{3.585})$ Clifford+$T$ gates without using any ancillae.
\end{lemma}
\begin{proof}
    Follows from the following circuit that generalises Eq.~\eqref{eq:two-controlled-Xplus1}:
\begin{equation}
    \input{gentwoxp1.tikz}
    \label{eq:gentwoxp1}
\end{equation}
As this requires two copies of $\ket{2}^{\otimes k}$-controlled $X_{01}$, the cost is asymptotically the same as in Lemma~\ref{lem:many-controlled-X01}.
\end{proof}

Note that we can change what the unitary is controlled on by conjugating any number of control qutrits by the appropriate qutrit $X$ gate in $\{X_{+1},X_{-1},X_{01},X_{02},X_{12}\}$. Successive application of tritstring-controlled $X$ gates leads to selective control on all possible combinations of tritstrings. We hence have the following.

\begin{corollary}\label{cor:anymct}
    We can construct any generalisation of a multiple-controlled Toffoli to the qutrit setting in Clifford+$T$ unitarily without ancillae.
\end{corollary}

\subsection{Building many controlled Clifford+\texorpdfstring{$T$}{T} gates}

Now we have all the tools to build the remaining controlled Clifford+$T$ gates: CX, Hadamard, $S$ and $T$.

\begin{lemma}
    The $\ket{2}^{\otimes k}$-controlled CX gate can be constructed unitarily without ancillae using a polynomial number of Clifford+$T$ gates in $k$.
\end{lemma}
\begin{proof}
    Just consider the following construction, which can be obtained from applying Lemma~\ref{lemma:vvsqtrick} with $V = X_{-1}$:
    \begin{equation}
        \input{gentwocx.tikz}
    \end{equation}
\end{proof}

For the many-controlled $T$ gate we use the `square-root trick' of Eq.~\eqref{qubitgenmczphase}. Here to find the square root we can use the fact $(T^5)^2 = T$ as $T^9 = \mathbb{I}$, and hence $T^5$ acts like a $\sqrt{T}$ gate. Similarly, $T^4$ is like $\sqrt{T}^\dagger$.
It is then easily verified that we have the following construction for the controlled $T$ gate.
\begin{lemma}
    We can build the $\ket{2}^{\otimes k}$-controlled $T$ gate unitarily without ancillae using a polynomial number of Clifford+$T$ gates in $k$.
\end{lemma}
\begin{proof}
    Consider the following circuit identity:
    \begin{equation}\label{eq:gentwot}
        \input{gentwot.tikz}
    \end{equation}
    Its correctness follows because $X_{12}T^4X_{12} = T^5$.
\end{proof}
That we can implement the $\ket{2}$-controlled $T$ in Clifford+$T$ without ancillae is interesting as this is not possible in the qubit setting. With qubits, to construct the controlled $T$ gate we have to either employ the $\sqrt{T}$ gate, or a clean ancilla~\cite{BarencoA1995elementarygates, KitaevA1996abelianstab, AmyM2013depthopt, KliuchnikovV2013singlequbitcliffordplust, SelingerP2013tdepthone, ChoiK2018ctrlRn}.  Indeed, Kliuchnikov, Maslov, and Mosca showed that for qubit Clifford+$T$ a single ancilla is in fact necessary~\cite{KliuchnikovV2013singlequbitcliffordplust}.

To add controls to the $S$ gate, we extend our construction from Ref.~\cite{GlaudellA2022qutritmetaplecticsubset}:
\begin{lemma}
    \label{lemma:tcsp}
    The $\ket{2}^{\otimes k}$-controlled $S$ gate can be constructed unitarily without ancillae using a polynomial number of Clifford+$T$ gates in $k$.
\end{lemma}
\begin{proof}
    We straightforwardly generalise Eq.~\eqref{eq:controlled-S-circ}:
\begin{equation}
    \input{gentwosrel.tikz}
\end{equation}
\end{proof}

We will shortly show how to control the $H$ gate. To do this we will need to handle some controlled global phase for which we need the following results.
\begin{lemma}\label{lem:z22phase}
    The $\ket{2}^{\otimes k}$-controlled $Z(2,2) = \text{diag}(1,\omega^2,\omega^2)$ gate can be constructed up to a controlled global phase of $\zeta^2$.  This Clifford+$T$ construction is unitary and ancilla-free, with $T$-count polynomial in $k$.
\end{lemma}
\begin{proof}
    Use the following circuit:
    \begin{equation}
        \label{eq:tcz22phase}
        \input{tcz22phase.tikz}
    \end{equation}
    Its correctness can be verified by direct computation, or by commuting $S$ and $X_{02}$.
\end{proof}

With one borrowed ancilla, we can construct the $\ket{2}^{\otimes k}$-controlled $\zeta S=Z(0,1)$ gate:
\begin{lemma}\label{lem:Z01-gate-control}
    The $\ket{2}^{\otimes k}$-controlled $Z(0,1)$ gate can be constructed unitarily with one borrowed ancilla using only Clifford+T gates, with $T$-count polynomial in $k$.
\end{lemma}
\begin{proof}
    Use the following circuit:
    \begin{equation}
        \label{eq:gentwosborrow}
        \input{gentwosborrow.tikz}
    \end{equation}
    The correctness follows from the fact that $Z_{-1} X_{02} Z_{-1} X_{02} = \omega\mathbb{I}$, and from Eq.~\eqref{eq:controlledglobalphase}.
\end{proof}

This construction of the $\ket{2}^{\otimes k}$-controlled $\zeta S$ gate can be composed with a $\ket{2}^{\otimes k}$-controlled $S^\dagger$ gate, which we constructed in Lemma~\ref{lemma:tcsp}. 
The resulting unitary is then the controlled global phase of $\zeta$ gate;
\begin{corollary}
    We can construct the $\ket{2}^{\otimes k}$-controlled $\zeta \mathbb{I}$ gate unitarily without ancillae in Clifford+$T$. Equivalently, this gate is equal to the $\ket{2}^{\otimes (k-1)}$-controlled $Z(0,1/3)$ gate by Eq.~\eqref{eq:controlledglobalphase}.
\end{corollary}

\begin{lemma}
    The $\ket{2}^{\otimes k}$-controlled Hadamard gate can be implemented unitarily without ancillae using a polynomial number of Clifford+$T$ gates in $k$.
\end{lemma}
\begin{proof}
    Consider the following circuit:
    \begin{equation}
        \input{gentwosh.tikz}
    \end{equation}
    Each of these gates can be constructed unitarily in Clifford+$T$ without ancillaes and with a polynomial number of gates using Lemmas~\ref{lem:z22phase} and~\ref{lem:Z01-gate-control}.
    To see why it is correct, note that when the gates fire it implements a $\zeta^9 Z(2,2)X(2,2)Z(2,2)$ gate, and as $\zeta^9 = 1$ and $Z(2,2)X(2,2)Z(2,2)=H$ by Eq.~\eqref{eq:EUrule} the construction is indeed correct.
\end{proof}

\subsection{The main theorem}

We have now seen how to exactly construct unitarily and without ancillae the $\ket{2}^{\otimes k}$-controlled CX, $S$, $H$, $T$ and permutation gates. Each of our constructions used at most $O(k^{3.585})$ Clifford+$T$ gates. 

We hence have the following theorem.
\begin{theorem}\label{thm:main}
    Let $U$ be any $n$-qutrit Clifford+$T$ unitary consisting of $N$ CX, $S$, $H$ and $T$ gates. 
    Then there is a $n+k$-qutrit Clifford+$T$ circuit implementing the $\ket{2}^{\otimes k}$-controlled $U$ unitary using $O(N k^{3.585})$ gates.
\end{theorem}

As noted in Corollaries~\ref{cor:conj22}~and~\ref{cor:anymct}, by conjugating the control qutrits by certain Clifford operations, we can make the unitary controlled on every possible tritstring in the $Z$- or $X$-basis.
Furthermore, by repeated application of the decomposition on different control values, more specific definitions of ternary control, including that of Def.~\ref{def:controlled-gates} may be realized.

\section{Building trit permutation gates in Clifford+\texorpdfstring{$T$}{T}}\label{sec:permutations}
A \emph{ternary classical reversible} function is a bijective map $f:\{0,1,2\}^n\to \{0,1,2\}^n$.
Ternary classical reversible circuits have been well-studied with a variety of applications.
In contrast with irreversible logic, which necessarily dissipates energy in order to perform computation, reversible logic is of interest for energy efficient and sustainable computing.
Ternary classical reversible functions are of importance to quantum algorithms involving oracles, which are implementations of classical functions as quantum gates.
One reason to expect comparative advantage of qutrits over qubits for this application is that it is impossible to build ancilla-free qubit Clifford+$T$ Toffolis with $3$ or more controls~\cite{KliuchnikovV2013singlequbitcliffordplust}.  Therefore, ancilla-free implementations of classical reversible functions as qubit Clifford+$T$ unitaries are limited to bitstrings of length $n \leq 3$.

However, the decomposition of ternary classical reversible circuits into a fault-tolerant quantum gate set such as Clifford+$T$ has not been explicitly presented before.
Given a ternary classical reversible function $f$ on $n$ trits, our goal will be to find an $n$-qutrit Clifford+$T$ unitary $U$ that implements that function on its $Z$ basis states: $U\ket{x_1,\ldots, x_n} = \ket{f(x_1,\ldots,x_n)}$. Additionally, we want this construction to require the fewest number of $T$ gates.

Let's first note that if we put some previously appeared papers together, that there is in fact a procedure to build arbitrary ternary reversible functions as ancilla-free Clifford+$T$ circuits, although this has not been noted before.
Namely, Fan~\emph{et.~al.}~\cite{FanF2015reversiblequtrit} proved that a gate set consisting of only the single-qutrit gates $X_{01}$, $X_{02}$, and $X_{12}$; and the two-qutrit $\ket{1}$-controlled $X_{+1}$ gate; can implement all the ternary classical reversible circuits without ancillae. They do this by proving that the $\ket{2}^{\otimes (n-1)}$-controlled $X$ gates needed to build reversible functions can be broken down in terms of single-qutrit $X$ gates and the $\ket{1}$-controlled $X_{+1}$ gate (although this is not done in an explicit manner, and actually extracting the concrete circuits to do so is cumbersome).

Separately, Bocharov, Roetteler, and Svore~\cite[Figure 6]{BocharovA2017ternaryshor} constructed the $\ket{1}$-controlled $X_{+1}$ gate as a two-qutrit Clifford+$T$ unitary. Combining these results then yields ancilla-free Clifford+$T$ implementations of any ternary classical reversible function.

Various past works have compared optimisation procedures for specific benchmark ternary classical reversible circuits~\cite{KhanM2004ternarysynthesisgenetic, MandalSBternarysynthesisprojection, BasuSsynthesisternaryreversible, Rani_Kole_Datta_Chakrabarty_2016, KoleA2017synthesisternaryreversible, RaniPM2020grouptheoreticternaryreversible}.
Rather than focus on optimising individual circuits, we will instead focus on providing an asymptotic upper bound of the Clifford+$T$ gate count needed to construct an arbitrary ternary classical reversible circuit. This same upper bound then also caps the asymptotic $T$-count of these previous decompositions.
We hope that our algorithm can then serve as a baseline for circuit complexity of ternary classical reversible circuit synthesis, to which empirical gate counts for specific constructions can be compared.

As was observed in Ref.~\cite{YangG2006algreversible}, we can view a reversible classical function on $n$ trits as just a permutation of the set $\{0,1,2\}^n$. The following classic result shows that any permutation is generated by just the $2$-cycles.
\begin{definition}
   Let $S_k$ be a symmetric group of symbols $\{d_1,d_2,...,d_k\}$, then $(d_{i_1}, d_{i_2},...,d_{i_j})$ is called a `j'-cycle, where $j \leq k, 1 \leq i_1,i_2,...,i_j \leq k$.
\end{definition}
\begin{lemma}\label{lem:2cycles}
    The $2$-cycles generate all permutations.
\end{lemma}
\begin{proof}
    Any permutation can be written as product of some disjoint cycles.  So we only need to show that every cycle $(d_1,d_2,...,d_k)$ can be expressed as a product of some $2$-cycles. When we have a $k$-cycle we can reduce it to a $(k-1)$-cycle:
    \begin{equation}\label{eq:productcycles}
        (d_1,d_2,...,d_k) = (d_1,d_2)(d_1,d_3,...,d_k)
    \end{equation}
    By repeating this equation, each cycle can then be reduced to a $2$ cycle.
\end{proof}
It hence suffices to show how to implement an arbitrary $2$-cycle on the set $\{0,1,2\}^n$.
Such a $2$-cycle is called a \emph{two-level axial reflection} in Ref.~\cite[Def.~1]{BocharovA2017ternaryshor}. That is, a two-level axial reflection is a qutrit operation which permutes two tritstrings, and acts as the identity on all other tritstrings.
\begin{definition}\label{def:twolevelaxial}
    A \emph{two-level axial reflection} is a gate
    \begin{equation}
        \tau_{\ket{j},\ket{k}} = \mathbb{I}^{\otimes n} - \ket{j}\bra{j} - \ket{k}\bra{k} + \ket{j}\bra{k} + \ket{k}\bra{j}
    \end{equation}
    where $\ket{j},\ket{k}$ are two different standard $n$-qudit basis vectors.
\end{definition}

In Ref.~\cite[Section 4]{BocharovA2016topologicalcompilation}, later improved by a constant factor in Ref.~\cite{BocharovA2016optimalitymetaplectic}, they presented a method for \emph{approximate} synthesis of any $n$-qutrit two-level axial reflection in the Clifford+$R$ gate set, where $R = \text{diag}(1,1,-1)$ is the non-Clifford gate for attaining universality, with asymptotic $R$-count $5 \text{log}_3(1/3\epsilon) + O(\text{log}(\text{log}(1/\epsilon)))$, such that $c\ket{0}^{\otimes n}$ approximates $\ket{u}$ to precision $\epsilon / (2\sqrt(2))$.

We show that the same operation, an $n$-qutrit two-level axial reflection, can be \emph{exactly} synthesized as an ancilla-free Clifford+$T$ unitary with a $T$-count that scales asymptotically as $O(n^{3.585})$. This is because we show we can construct it using a single instance of a qutrit Toffoli controlled on $n-1$ wires which requires $O(n^{3.585})$ gates. 

\begin{proposition}
    Let $\vec{a} = (a_1,...,a_n)$ and $\vec{b} = (b_1,...,b_n)$ be any two tritstrings of length $n$.
    Then we can exactly implement the two-level axial reflection on $\vec{a}$ and $\vec{b}$ as an ancilla-free $n$-qutrit Clifford+$T$ unitary with $T$-count~$O(n^{3.585})$.
\end{proposition}
\begin{proof}
    We assume $\vec{a} \neq \vec{b}$, or the 2-cycle would just be the identity operation on all inputs.
    As $\vec{a}$ and $\vec{b}$ differ, they must differ by at least one character.  Without loss of generality suppose that $a_n \neq b_n$.
    Consider the following circuit:
    \begin{equation}\label{eq:permutetwotritstrings}
        \input{permutetwotritstrings.tikz}
    \end{equation}
    Here the circles denote controls on the value of an $a_j$ or $b_j$, which control whether a $X_{a_j,b_j}$ operation is applied (which is identity if $a_j = b_j$). 
    Hence, the gate in Step~2 is a many-controlled $X_{a_n,b_n}$ gate, which we know how to build by Lemma~\ref{lem:many-controlled-X01} using $O(n^{3.585})$ gates.
    We conjugate this gate, in Steps~1~and~3, by $n-1$ gates that are each Clifford equivalent to the $\ket{2}$-controlled $X_{12}$ gate, and hence require $O(n)$ gates to implement. 
    
    This circuit indeed implements the $(\vec{a},\vec{b})$ 2-cycle, which we can see through enumerating the possible input cases.
    \begin{itemize}
        \item When the input is $\vec{a}$: Only steps~2~and~3 fire (as $b_n\neq a_n$), outputting $\vec{b}$.
        \item When the input is $\vec{b}$: Steps~1~and~2 fire, outputting $\vec{a}$.
    \end{itemize}
    Observe that when Step~2 does not fire, Steps~1~and~3 always combine to the identity gate.
    Therefore, we only need to consider the remaining cases where Step~2 does fire.
    \begin{itemize}
        \item When both Steps~1~and~2 fired: The input had to have been $\vec{b}$.
        \item When Step~2 fired, but Step~1 didn't fire: Either the input was $\vec{a}$, or the last input character was neither $a_n$ nor $b_n$ in which case the overall operation is the identity.
    \end{itemize}
    Therefore, the circuit in Eq.~\eqref{eq:permutetwotritstrings} maps $\vec{a}$ to $\vec{b}$, $\vec{b}$ to $\vec{a}$, and is identity on all other tritstrings.
    Lastly, the $T$-count is asymptotically that of the gate in Step~2, which is $O(n^{3.585})$.
\end{proof}

We note that this decomposition can readily be generalized to any qudit dimension provided we can construct single-qudit $X$ permutation gates controlled on an arbitrary dit string. 
\begin{theorem}\label{thm:classical-synthesis}
    For any ternary classical reversible function $f:\{0,1,2\}^n\to \{0,1,2\}^n$ on $n$ trits, we can construct an $n$-qutrit ancilla-free Clifford+$T$ circuit which exactly implements it, with $T$-count $O(3^n n^{3.585})$.
\end{theorem}
\begin{proof}
    We view $f$ as a permutation of size $3^n$. This permutation consists of cycles, each of which can be decomposed into 2-cycles using Lemma~\ref{lem:2cycles}. This full decomposition requires at most $3^n-1$ 2-cycles. Implementing each of these 2-cycles requires $O(n^{3.585})$ $T$ gates.  Therefore, the asymptotic $T$-count of the overall construction is $O(3^n n^{3.585})$.
\end{proof}
At first glance the number of gates needed here may seem expensive, but as noted in Ref.~\cite{YangG2006algreversible}, the asymptotic scaling we find is still ``exponentially lower than the complexity of a breadth-first-search synthesis algorithm''.
Additionally, with a simple counting argument we can show our result is within a polynomial factor of the optimal number.
\begin{proposition}
    There exist ternary classical reversible functions $f:\{0,1,2\}^n\to \{0,1,2\}^n$ that require at least $O(n3^n/\log n)$ Clifford+$T$ gates to construct.
\end{proposition}
\begin{proof}
    We consider a gate set of CX, $S$, $T$, and $H$ gates. Taking into account qutrit positioning there are then $n(n-1) + 3n$ different gates, so that using $N$ of these gates, we can construct at most $(n(n-1)+3n)^N \leq (2n^2)^N = 2^Nn^{2N}$ different circuits. There are exactly $(3^n)!$ different classical reversible trit functions on $n$ trits (where $k!$ denotes the factorial of $k$). In order to write down every such permutation we must hence have a number of gates $N$ such that at least $2^Nn^{2N} \geq (3^n)!$. Taking the logarithm on both sides and using $\log(k!) \geq \frac12 k\log k$ we can rewrite this inequality to $N\log 2 + 2N\log n \geq \frac12 3^n\cdot n\log 3$. Factoring out $N$ gives $N\geq \frac{\log 3}{2} \frac{n3^n}{\log 2 + 2\log n} \geq \frac{\log 3}{6} \frac{n3^n}{\log n}$ showing that we must have $N=O(n3^n/\log n)$.
\end{proof}
We believe it might be possible to improve the implementation of the $X_{01}$ gate with $n$ controls to require just $O(n)$ gates, in which case the construction of Theorem~\ref{thm:classical-synthesis} would require $O(n 3^n)$ gates, making it optimal to within a logarithmic factor.

Although our construction resembles that for qubits in Ref.~\cite{YangG2006algreversible}, their construction uses $O(n 2^n)$ multiple-controlled $X$ gates (each with $n-1$ controls).  In contrast, our construction requires $O(3^n)$ of this gate's qutrit equivalent.  This factor of $n$ improvement in asymptotic circuit complexity can be applied to the qubit setting of Ref.~\cite{YangG2006algreversible} as well, resulting in a more efficient construction.
Additionally, our observation that only a single ternary $(n-1)$-controlled Toffoli is needed to implement the two-level axial reflection can be used to improve the algorithm of Fan~\emph{et.~al.}~\cite{FanF2015reversiblequtrit} as well.

\section{Conclusion}\label{sec:conclusion}

We have shown how to construct any many-controlled qutrit Clifford+$T$ unitary, using just Clifford+$T$ gates and without using ancillae. Our construction uses $O(k^{3.585})$ gates in the number of controls $k$.
Using our results we have shown how any classical permutation on $n$ trits can be realised as an $n$-qutrit ancilla-free Clifford+$T$ unitary circuit with $O(3^n n^{3.585})$ gates.

We suspect that the $O(k^{3.585})$ scaling is not optimal. In future work we would like to find better ways to decompose the many controlled $X_{+1}$ gate into lower-controlled gates, using the fact that after the first iteration of the decomposition we have many borrowed ancillae available, which would possibly be used to lead to better asymptotic scaling.  In particular, we would like to see whether the linear $T$-count construction of the qubit $n$-controlled Toffoli construction where $n-2$ borrowed ancilla are available from Ref.~\cite[Lemma~7.2]{BarencoA1995elementarygates} can be adapted to qutrits. 
Improvements in this scaling will directly lead to improvements in the Clifford+$T$ synthesis of reversible trit functions of Theorem~\ref{thm:classical-synthesis} and will bring it closer to the theoretical optimum. It would also directly improve the decompositions using the techniques of for instance Refs.~\cite{YangG2006algreversible, KoleD2010reversibledfsbfs, KamalikaD2012grouptheoryreversible, RaniPM2020grouptheoreticternaryreversible}.
It would also be interesting to find lower bounds on the $T$-count of our constructions using techniques extended from the qubit setting~\cite{SongG2004optsimptoffoli, ShendeVopttoffolicx, GiuliaM2020enumoptclassify, GossetD2014opttoffolit, MaslovD2016optasymptoticNCTreversible, HowardM2017resourcetheorymagic, BeverlandM2020lowerboundnonclifford, MoscaM2021polytimespacetcount}.

Our results pave the way to a full characterisation of the unitaries that can be constructed over the qutrit Clifford+$T$ gate set.
We conjecture that, as in the qubit case~\cite{GilesB2013multiqubitcliffordplustsynthesis}, any qutrit unitary with entries in the number ring generated by the Clifford+$T$ gate set can be exactly synthesised over Clifford+$T$.

Finally, we aim to use our results to emulate qubit logic circuits on qutrits. Work in this area has already shown to lead to several benefits~\cite{GokhaleP2019asymptotic}, so it will be interesting to to identify where more asymptotic improvements for qubit computation in the fault-tolerant regime can be made.

\paragraph{Acknowledgments} The authors wish to thank Andrew Glaudell and Neil J.~Ross for discussions regarding the consequences of our results and Andrew Glaudell specifically for pointing out Eq.~\eqref{eq:gentwot}.  We additionally wish to thank Shuxiang Cao and Razin Shaikh for assistance in preparing the figures in an early draft of this paper. JvdW is supported by a NWO Rubicon personal fellowship. LY is supported by an Oxford - Basil Reeve Graduate Scholarship at Oriel College with the Clarendon Fund.

\bibliographystyle{splncs04}
\bibliography{main}

\end{document}

%% file: zx.tikzstyles
\tikzstyle{box}=[shape=rectangle, text height=1.5ex, text depth=0.25ex, yshift=0.5mm, fill=white, draw=black, minimum height=5mm, yshift=-0.5mm, minimum width=5mm, font={\small}]
\tikzstyle{Z dot}=[inner sep=0mm, minimum size=2mm, shape=circle, draw=black, fill={rgb,255: red,216; green,248; blue,216}, tikzit fill={rgb,255: red,216; green,248; blue,216}]
\tikzstyle{Z phase dot}=[minimum size=4.75mm, font={\footnotesize}, shape=rectangle, rounded corners=1.9mm, inner sep=0.1mm, outer sep=-2mm, scale=0.8, tikzit shape=circle, draw=black, fill={rgb,255: red,216; green,248; blue,216}, tikzit draw=blue, tikzit fill={rgb,255: red,216; green,248; blue,216}]
\tikzstyle{X dot}=[Z dot, shape=circle, draw=black, fill={rgb,255: red,232; green,165; blue,165}, tikzit fill={rgb,255: red,232; green,165; blue,165}]
\tikzstyle{X phase dot}=[Z phase dot, tikzit shape=circle, tikzit fill={rgb,255: red,232; green,165; blue,165}, fill={rgb,255: red,232; green,165; blue,165}, font={\footnotesize}, tikzit draw=blue]
\tikzstyle{XD dot}=[shape=XDdot, inner sep=2pt, draw=black]
\tikzstyle{XD phase dot}=[shape=XDdotphase, minimum size=4.75mm, font={\footnotesize}, inner sep=.1mm, outer sep=0mm, scale=0.8, tikzit shape=circle, rounded corners=1.9mm, draw=black]
\tikzstyle{zn}=[shape=zn, tikzit draw=black, draw=black, inner sep=2pt]
\tikzstyle{hadamard}=[fill={rgb,255: red,140; green,220; blue,248}, draw=black, shape=rectangle, inner sep=0.6mm, minimum height=1.5mm, minimum width=1.5mm, xslant=0.5]
\tikzstyle{hz}=[hadamard, fill={rgb,255: red,216; green,248; blue,216}, shape=rectangle, tikzit fill={rgb,255: red,216; green,248; blue,216}, minimum height=2 mm, minimum width=1.25 mm, tikzit draw=black]
\tikzstyle{hx}=[hadamard, fill={rgb,255: red,232; green,165; blue,165}, shape=rectangle, tikzit fill={rgb,255: red,232; green,165; blue,165}, minimum height=2 mm, minimum width=1.25 mm, tikzit draw=black]
\tikzstyle{vertex}=[inner sep=0mm, minimum size=1mm, shape=circle, draw=black, fill=black]
\tikzstyle{vertex set}=[inner sep=0mm, minimum size=1mm, shape=circle, draw=black, fill=white, font={\footnotesize\boldmath}]
\tikzstyle{meter}=[draw, fill=white, minimum width=2em, minimum height=1.5em, rectangle, path picture={\draw ([shift={(.1,.24)}]path picture bounding box.south west) to[bend left=50] ([shift={(-.1,.24)}]path picture bounding box.south east);\draw[-{Latex[scale=0.6]}] ([shift={(0,.1)}]path picture bounding box.south) -- ([shift={(.3,-.1)}]path picture bounding box.north);}, tikzit shape=rectangle]
\tikzstyle{white dot}=[Z dot]
\tikzstyle{gray dot}=[X dot]
\tikzstyle{white phase dot}=[Z phase dot]
\tikzstyle{gray phase dot}=[X phase dot]
\tikzstyle{red ket}=[fill={rgb,255: red,232; green,165; blue,165}, draw=black, shape=isosceles triangle, tikzit fill={rgb,255: red,232; green,165; blue,165}, tikzit draw=black, inner sep=0 mm, outer sep=2 mm]
\tikzstyle{tiny none}=[none, font={\tiny}]
\tikzstyle{green ket}=[fill={rgb,255: red,216; green,248; blue,216}, draw=black, shape=isosceles triangle, tikzit fill={rgb,255: red,216; green,248; blue,216}, tikzit draw=black, inner sep=0 mm, outer sep=2 mm]
\tikzstyle{filament}=[hadamard, fill=yellow, draw=none, minimum height=0.01mm]
\tikzstyle{unit circle}=[shape=circle, minimum size=42.5 mm, fill=none, draw={rgb,255: red,223; green,223; blue,223}, tikzit draw={rgb,255: red,223; green,223; blue,223}]
\tikzstyle{new style 0}=[shape=ellipse, minimum height=10 mm, minimum width=100 mm, fill=white, draw=black]
\tikzstyle{gate}=[box, minimum height=10mm, minimum width=10mm]
\tikzstyle{2 control}=[vertex set, draw=blue, inner sep=0.5pt]
\tikzstyle{1 control}=[2 control, draw=red]
\tikzstyle{0 control}=[2 control, draw=black]
\tikzstyle{small dot}=[vertex, minimum size=1 mm, draw=black, tikzit draw=black, tikzit fill=black, tikzit shape=circle]
\tikzstyle{tallbox}=[box, minimum height=12mm]
\tikzstyle{targ}=[vertex set, minimum size=0.5mm, inner sep=-0.5mm, tikzit shape=circle, shape=circle, tikzit draw=black]
\tikzstyle{hadamardbox}=[hadamard, xslant=0]

\tikzstyle{directedarrow}=[draw={rgb,255: red,223; green,223; blue,223}, ->, tikzit draw={rgb,255: red,223; green,223; blue,223}, line width=1 pt]
\tikzstyle{simple}=[-]
\tikzstyle{hadamard edge}=[-, color={rgb,255: red,0; green,100; blue,248}, dashed, dash pattern=on 2pt off 0.7pt, tikzit draw={rgb,255: red,0; green,100; blue,248}]
\tikzstyle{brace edge}=[-, tikzit draw=blue, decorate, decoration={brace,amplitude=1mm,raise=-1mm}]
\tikzstyle{gray}=[-, draw={rgb,255: red,223; green,223; blue,223}, line width=1 pt]
\tikzstyle{arrow}=[<-, draw={rgb,255: red,128; green,128; blue,128}]
\tikzstyle{double-arrow}=[draw={rgb,255: red,128; green,128; blue,128}, <->]
\tikzstyle{dashed edge}=[-, dashed, dash pattern=on 2pt off 0.5pt, draw=black]
\tikzstyle{diredge}=[->]
\tikzstyle{double edge}=[-, double, shorten <=-1mm, shorten >=-1mm, double distance=2pt]
\tikzstyle{thin}=[-, line width=0.05mm]
\tikzstyle{thin gray}=[-, draw={rgb,255: red,223; green,223; blue,223}, line width=0.05mm]
\tikzstyle{less thin}=[-, line width=0.1mm]
\tikzstyle{dashed gray edge}=[-, dashed edge, draw={rgb,255: red,128; green,128; blue,128}]
\tikzstyle{light right directed arrow}=[->, directedarrow, draw={rgb,255: red,223; green,223; blue,223}, line width=0.2mm]
\tikzstyle{diredge0.3}=[->, line width=0.3 mm]
\tikzstyle{less thin gray}=[-, draw={rgb,255: red,223; green,223; blue,223}]
\tikzstyle{dashed thin purple}=[-, dashed, line width=0.1mm, draw={rgb,255: red,128; green,106; blue,219}]
\tikzstyle{hadamardedge}=[-, color={rgb,255: red,100; green,200; blue,248}, dashed, dash pattern=on 2pt off 0.7pt, tikzit draw={rgb,255: red,120; green,220; blue,248}]
\tikzstyle{line0.3}=[-, line width=0.3mm]
\tikzstyle{light blue line}=[-, color={rgb,255: red,100; green,200; blue,248}, tikzit draw={rgb,255: red,100; green,200; blue,248}]
\tikzstyle{pink0.2}=[-, color={rgb,255: red,248; green,100; blue,200}, line width=0.2mm, tikzit draw={rgb,255: red,248; green,100; blue,200}]

%% file: figures/zczp1.tikz
\begin{tikzpicture}
	\begin{pgfonlayer}{nodelayer}
		\node [style=0 control] (0) at (-9, 1) {\textsf{0}};
		\node [style=box] (1) at (-9, -1) {$Z$};
		\node [style=none] (3) at (-7, 0) {=};
		\node [style=box] (6) at (-4, -1) {$T$};
		\node [style=box] (7) at (-2, -1) {$X$};
		\node [style=box] (8) at (0, -1) {$T$};
		\node [style=box] (9) at (2, -1) {$X$};
		\node [style=box] (10) at (4, -1) {$T$};
		\node [style=box] (11) at (6, -1) {$X$};
		\node [style=none] (14) at (8.25, -1) {};
		\node [style=vertex set] (15) at (6, 1.025) {$\Lambda$};
		\node [style=vertex set] (16) at (-2, 1.025) {$\Lambda$};
		\node [style=vertex set] (17) at (2, 1.025) {$\Lambda$};
		\node [style=none] (19) at (-6, 1) {};
		\node [style=none] (20) at (-6, -1) {};
		\node [style=none] (21) at (8.25, 1) {};
		\node [style=none] (26) at (-10, 1) {};
		\node [style=none] (27) at (-10, -1) {};
		\node [style=none] (28) at (-8, 1) {};
		\node [style=none] (29) at (-8, -1) {};
	\end{pgfonlayer}
	\begin{pgfonlayer}{edgelayer}
		\draw [in=90, out=-90] (0) to (1);
		\draw (6) to (7);
		\draw (7) to (8);
		\draw (8) to (9);
		\draw (9) to (10);
		\draw (10) to (11);
		\draw (16) to (7);
		\draw (17) to (9);
		\draw (15) to (11);
		\draw (16) to (17);
		\draw (17) to (15);
		\draw (26.center) to (0);
		\draw (0) to (28.center);
		\draw (27.center) to (1);
		\draw (1) to (29.center);
		\draw [style=simple] (20.center) to (6);
		\draw [style=simple] (14.center) to (11);
		\draw [style=simple] (21.center) to (15);
		\draw [style=simple] (16) to (19.center);
	\end{pgfonlayer}
\end{tikzpicture}

%% file: figures/tcd01.tikz
\begin{tikzpicture}
	\begin{pgfonlayer}{nodelayer}
		\node [style=2 control] (0) at (-13.25, 1) {\textcolor{blue}{\textsf{2}}};
		\node [style=box] (1) at (-13.25, -1) {$X_{01}$};
		\node [style=none] (3) at (-11, 0) {=};
		\node [style=box] (4) at (-6, 1) {$X_{-1}$};
		\node [style=box] (6) at (-3.5, 1) {$X_{+1}$};
		\node [style=1 control] (7) at (-1.5, 1) {\textcolor{red}{\textsf{1}}};
		\node [style=box] (8) at (0.5, 1) {$X_{+1}$};
		\node [style=1 control] (9) at (2.5, 1) {\textcolor{red}{\textsf{1}}};
		\node [style=box] (10) at (4.5, 1) {$X_{+1}$};
		\node [style=none] (11) at (8.5, 1) {};
		\node [style=box] (12) at (11.75, -1) {$X_{12}$};
		\node [style=1 control] (13) at (-3.5, -1) {\textcolor{red}{\textsf{1}}};
		\node [style=1 control] (14) at (0.5, -1) {\textcolor{red}{\textsf{1}}};
		\node [style=1 control] (15) at (4.5, -1) {\textcolor{red}{\textsf{1}}};
		\node [style=none] (16) at (-10, 1) {};
		\node [style=none] (17) at (-10, -1) {};
		\node [style=none] (18) at (13.75, 1) {};
		\node [style=none] (19) at (13.75, -1) {};
		\node [style=none] (20) at (-14.5, 1) {};
		\node [style=none] (21) at (-14.5, -1) {};
		\node [style=none] (22) at (-12, 1) {};
		\node [style=none] (23) at (-12, -1) {};
		\node [style=none] (24) at (8.5, -1) {};
		\node [style=box] (25) at (-1.5, -1) {$X_{+1}$};
		\node [style=box] (26) at (2.5, -1) {$X_{+1}$};
		\node [style=vertex set] (27) at (-6, -1) {$\Lambda$};
		\node [style=box] (28) at (-8, -1) {$X_{12}$};
		\node [style=box] (33) at (9.75, 1) {$X_{+1}$};
		\node [style=vertex set] (34) at (9.75, -1) {$\Lambda$};
	\end{pgfonlayer}
	\begin{pgfonlayer}{edgelayer}
		\draw (0) to (1);
		\draw (6) to (7);
		\draw (7) to (8);
		\draw (8) to (9);
		\draw (9) to (10);
		\draw (8) to (14);
		\draw (6) to (13);
		\draw (10) to (15);
		\draw (12) to (19.center);
		\draw (20.center) to (0);
		\draw (0) to (22.center);
		\draw (21.center) to (1);
		\draw (1) to (23.center);
		\draw [in=-15, out=180] (24.center) to (10);
		\draw [in=-180, out=0] (15) to (11.center);
		\draw (13) to (25);
		\draw (25) to (14);
		\draw (14) to (26);
		\draw (26) to (15);
		\draw (9) to (26);
		\draw (25) to (7);
		\draw (27) to (13);
		\draw (17.center) to (28);
		\draw (28) to (27);
		\draw (4) to (6);
		\draw (16.center) to (4);
		\draw (4) to (27);
		\draw [style=simple] (24.center) to (34);
		\draw [style=simple] (34) to (12);
		\draw [style=simple] (11.center) to (33);
		\draw [style=simple] (33) to (18.center);
		\draw [style=simple] (33) to (34);
	\end{pgfonlayer}
\end{tikzpicture}

%% file: figures/controlledglobalphase.tikz
\begin{tikzpicture}
	\begin{pgfonlayer}{nodelayer}
		\node [style=2 control] (0) at (-11, -0.75) {\textcolor{blue}{\textsf{2}}};
		\node [style=none] (1) at (-11, 2.5) {};
		\node [style=none] (2) at (-11, 0) {};
		\node [style=none] (3) at (-11, 1.375) {$\vdots$};
		\node [style=2 control] (4) at (-11, 3.25) {\textcolor{blue}{\textsf{2}}};
		\node [style=box] (5) at (-11, -2.75) {$e^{i \frac{2}{3} \pi \gamma} \mathbb{I}$};
		\node [style=none] (6) at (-13.25, -2.75) {};
		\node [style=none] (7) at (-13.25, 3.25) {};
		\node [style=none] (8) at (-13.25, -0.75) {};
		\node [style=none] (9) at (-8.75, 3.25) {};
		\node [style=none] (10) at (-8.75, -0.75) {};
		\node [style=none] (11) at (-8.75, -2.75) {};
		\node [style=none] (12) at (-9.5, -3) {};
		\node [style=none] (13) at (-9.25, -2.5) {};
		\node [style=none] (14) at (-9.5, -3.25) {$m$};
		\node [style=none] (15) at (-12.75, -3) {};
		\node [style=none] (16) at (-12.5, -2.5) {};
		\node [style=none] (17) at (-12.75, -3.25) {$m$};
		\node [style=none] (18) at (-13.75, 3.5) {};
		\node [style=none] (19) at (-13.75, -2) {};
		\node [style=none] (20) at (-14.25, 0.75) {};
		\node [style=none] (21) at (-14.75, 0.625) {};
		\node [style=none] (22) at (-14.75, 0.625) {$n$};
		\node [style=2 control] (23) at (-11, -1.75) {\textcolor{blue}{\textsf{2}}};
		\node [style=none] (25) at (-13.25, -1.75) {};
		\node [style=none] (26) at (-8.75, -1.75) {};
		\node [style=2 control] (27) at (-3.75, -0.75) {\textcolor{blue}{\textsf{2}}};
		\node [style=none] (28) at (-3.75, 2.5) {};
		\node [style=none] (29) at (-3.75, 0) {};
		\node [style=none] (30) at (-3.75, 1.375) {$\vdots$};
		\node [style=2 control] (31) at (-3.75, 3.25) {\textcolor{blue}{\textsf{2}}};
		\node [style=none] (33) at (-5.25, -2.75) {};
		\node [style=none] (34) at (-5.25, 3.25) {};
		\node [style=none] (35) at (-5.25, -0.75) {};
		\node [style=none] (36) at (-2.25, 3.25) {};
		\node [style=none] (37) at (-2.25, -0.75) {};
		\node [style=none] (38) at (-2.25, -2.75) {};
		\node [style=none] (42) at (-3.875, -3) {};
		\node [style=none] (43) at (-3.625, -2.5) {};
		\node [style=none] (44) at (-3.875, -3.25) {$m$};
		\node [style=none] (45) at (-5.75, 3.5) {};
		\node [style=none] (46) at (-5.75, -2) {};
		\node [style=none] (47) at (-6.25, 0.75) {};
		\node [style=none] (48) at (-6.75, 0.625) {};
		\node [style=none] (49) at (-6.75, 0.625) {$n$};
		\node [style=box] (50) at (-3.75, -1.75) {$Z(0,\gamma)$};
		\node [style=none] (51) at (-5.25, -1.75) {};
		\node [style=none] (52) at (-2.25, -1.75) {};
		\node [style=none] (53) at (-7.5, 0) {=};
		\node [style=2 control] (54) at (2.75, 1.75) {\textcolor{blue}{\textsf{2}}};
		\node [style=none] (55) at (2.75, 2.85) {};
		\node [style=none] (56) at (2.75, 2.15) {};
		\node [style=none] (57) at (2.75, 2.7) {$\vdots$};
		\node [style=2 control] (58) at (2.75, 3.25) {\textcolor{blue}{\textsf{2}}};
		\node [style=none] (59) at (1.25, -2.75) {};
		\node [style=none] (60) at (1.25, 3.25) {};
		\node [style=none] (61) at (1.25, 1.75) {};
		\node [style=none] (62) at (4.25, 3.25) {};
		\node [style=none] (63) at (4.25, 1.75) {};
		\node [style=none] (64) at (4.25, -2.75) {};
		\node [style=none] (65) at (2.625, -3) {};
		\node [style=none] (66) at (2.875, -2.5) {};
		\node [style=none] (67) at (2.625, -3.25) {$m$};
		\node [style=none] (68) at (0.75, 3.5) {};
		\node [style=none] (69) at (0.75, -2) {};
		\node [style=none] (70) at (0.25, 0.75) {};
		\node [style=none] (71) at (-0.25, 0.625) {};
		\node [style=none] (72) at (-0.25, 0.625) {$n$};
		\node [style=box] (73) at (2.75, 0.75) {$Z(0,\gamma)$};
		\node [style=none] (74) at (1.25, 0.75) {};
		\node [style=none] (75) at (4.25, 0.75) {};
		\node [style=none] (76) at (-1, 0) {=};
		\node [style=2 control] (77) at (2.75, -1.75) {\textcolor{blue}{\textsf{2}}};
		\node [style=2 control] (81) at (2.75, -0.25) {\textcolor{blue}{\textsf{2}}};
		\node [style=none] (82) at (1.25, -0.25) {};
		\node [style=none] (83) at (1.25, -1.75) {};
		\node [style=none] (84) at (4.25, -0.25) {};
		\node [style=none] (85) at (4.25, -1.75) {};
		\node [style=none] (86) at (2.75, -0.65) {};
		\node [style=none] (87) at (2.75, -1.35) {};
		\node [style=none] (88) at (2.75, -0.8) {$\vdots$};
	\end{pgfonlayer}
	\begin{pgfonlayer}{edgelayer}
		\draw (2.center) to (0);
		\draw (1.center) to (4);
		\draw (6.center) to (5);
		\draw (7.center) to (4);
		\draw (8.center) to (0);
		\draw (5) to (11.center);
		\draw (0) to (10.center);
		\draw (4) to (9.center);
		\draw (12.center) to (13.center);
		\draw (15.center) to (16.center);
		\draw [in=180, out=0, looseness=0.75] (20.center) to (19.center);
		\draw [in=-180, out=0, looseness=0.75] (20.center) to (18.center);
		\draw (25.center) to (23);
		\draw (23) to (26.center);
		\draw (0) to (23);
		\draw (23) to (5);
		\draw (29.center) to (27);
		\draw (28.center) to (31);
		\draw (34.center) to (31);
		\draw (35.center) to (27);
		\draw (27) to (37.center);
		\draw (31) to (36.center);
		\draw (42.center) to (43.center);
		\draw [in=180, out=0, looseness=0.75] (47.center) to (46.center);
		\draw [in=-180, out=0, looseness=0.75] (47.center) to (45.center);
		\draw [in=180, out=0] (51.center) to (50);
		\draw (50) to (52.center);
		\draw (27) to (50);
		\draw (33.center) to (38.center);
		\draw (56.center) to (54);
		\draw (55.center) to (58);
		\draw (60.center) to (58);
		\draw (61.center) to (54);
		\draw (54) to (63.center);
		\draw (58) to (62.center);
		\draw (65.center) to (66.center);
		\draw [in=180, out=0, looseness=0.75] (70.center) to (69.center);
		\draw [in=-180, out=0, looseness=0.75] (70.center) to (68.center);
		\draw (74.center) to (73);
		\draw (73) to (75.center);
		\draw (54) to (73);
		\draw (59.center) to (64.center);
		\draw (82.center) to (81);
		\draw (83.center) to (77);
		\draw (77) to (85.center);
		\draw (81) to (84.center);
		\draw (73) to (81);
		\draw [style=simple] (81) to (86.center);
		\draw [style=simple] (87.center) to (77);
	\end{pgfonlayer}
\end{tikzpicture}

%% file: figures/tcspdagphase.tikz
\begin{tikzpicture}
	\begin{pgfonlayer}{nodelayer}
		\node [style=none] (1) at (-5.5, -0.625) {};
		\node [style=none] (2) at (-5.5, 0.875) {};
		\node [style=none] (4) at (-1.5, -0.625) {};
		\node [style=none] (5) at (-1.5, 0.875) {};
		\node [style=2 control] (6) at (-3.5, 0.875) {\textcolor{blue}{\textsf{2}}};
		\node [style=box] (7) at (-3.5, -0.625) {$S$};
		\node [style=none] (13) at (0, 0) {$=$};
		\node [style=box] (47) at (3, -0.625) {$X_{01}$};
		\node [style=none] (48) at (1.5, 0.875) {};
		\node [style=2 control] (49) at (9, 0.875) {\textcolor{blue}{\textsf{2}}};
		\node [style=box] (50) at (9, -0.625) {$X_{+1}$};
		\node [style=none] (65) at (1.5, -0.625) {};
		\node [style=box] (66) at (5, -0.625) {$T^\dagger$};
		\node [style=box] (67) at (7, -0.625) {$X_{01}$};
		\node [style=box] (68) at (11, -0.625) {$X_{01}$};
		\node [style=2 control] (69) at (17, 0.875) {\textcolor{blue}{\textsf{2}}};
		\node [style=box] (70) at (17, -0.625) {$X_{-1}$};
		\node [style=box] (75) at (13, -0.625) {$T$};
		\node [style=box] (76) at (15, -0.625) {$X_{01}$};
		\node [style=none] (79) at (18.5, -0.625) {};
		\node [style=none] (199) at (18.5, 0.875) {};
	\end{pgfonlayer}
	\begin{pgfonlayer}{edgelayer}
		\draw (6) to (7);
		\draw (2.center) to (6);
		\draw (6) to (5.center);
		\draw (4.center) to (7);
		\draw (7) to (1.center);
		\draw (49) to (50);
		\draw (48.center) to (49);
		\draw (65.center) to (47);
		\draw (69) to (70);
		\draw (50) to (68);
		\draw (47) to (66);
		\draw (66) to (67);
		\draw (67) to (50);
		\draw (68) to (75);
		\draw (75) to (76);
		\draw (76) to (70);
		\draw (49) to (69);
		\draw [style=simple] (69) to (199.center);
		\draw [style=simple] (70) to (79.center);
	\end{pgfonlayer}
\end{tikzpicture}

%% file: figures/controlled-had-circ.tikz
\begin{tikzpicture}
	\begin{pgfonlayer}{nodelayer}
		\node [style=none] (1) at (-4.5, -0.625) {};
		\node [style=none] (2) at (-4.5, 0.875) {};
		\node [style=none] (4) at (-1.5, -0.625) {};
		\node [style=none] (5) at (-1.5, 0.875) {};
		\node [style=2 control] (6) at (-3, 0.875) {\textcolor{blue}{\textsf{2}}};
		\node [style=box] (7) at (-3, -0.625) {$H$};
		\node [style=none] (13) at (0, 0) {$=$};
		\node [style=box] (47) at (3.75, -0.625) {$X_{02}$};
		\node [style=none] (48) at (2.25, 0.875) {};
		\node [style=2 control] (49) at (5.25, 0.875) {\textcolor{blue}{\textsf{2}}};
		\node [style=box] (50) at (5.25, -0.625) {$S$};
		\node [style=none] (65) at (2.25, -0.625) {};
		\node [style=box] (66) at (8.75, -0.625) {$H^\dagger$};
		\node [style=none] (79) at (21.5, -0.625) {};
		\node [style=none] (199) at (21.5, 0.875) {};
		\node [style=box] (200) at (6.75, -0.625) {$X_{02}$};
		\node [style=box] (201) at (10.5, -0.625) {$X_{02}$};
		\node [style=2 control] (202) at (12, 0.875) {\textcolor{blue}{\textsf{2}}};
		\node [style=box] (203) at (12, -0.625) {$S$};
		\node [style=box] (204) at (13.5, -0.625) {$X_{02}$};
		\node [style=box] (205) at (17, -0.625) {$X_{02}$};
		\node [style=2 control] (206) at (18.5, 0.875) {\textcolor{blue}{\textsf{2}}};
		\node [style=box] (207) at (18.5, -0.625) {$S$};
		\node [style=box] (208) at (20, -0.625) {$X_{02}$};
		\node [style=box] (209) at (15, -0.625) {$H$};
		\node [style=box] (210) at (3.5, 0.875) {$S$};
		\node [style=none] (211) at (1.5, 0) {$\zeta$};
	\end{pgfonlayer}
	\begin{pgfonlayer}{edgelayer}
		\draw (6) to (7);
		\draw (2.center) to (6);
		\draw (6) to (5.center);
		\draw (4.center) to (7);
		\draw (7) to (1.center);
		\draw (49) to (50);
		\draw (48.center) to (49);
		\draw (65.center) to (47);
		\draw (47) to (66);
		\draw (202) to (203);
		\draw (206) to (207);
		\draw (66) to (79.center);
		\draw (199.center) to (49);
	\end{pgfonlayer}
\end{tikzpicture}

%% file: figures/c2c1vsqc2c2v.tikz
\begin{tikzpicture}
	\begin{pgfonlayer}{nodelayer}
		\node [style=none] (167) at (-6, 2) {};
		\node [style=none] (168) at (-6, 0.125) {};
		\node [style=none] (169) at (-6, -1.75) {};
		\node [style=box] (170) at (-7.25, -1.75) {$V^2$};
		\node [style=2 control] (171) at (-7.25, 2) {\textcolor{blue}{\textsf{2}}};
		\node [style=2 control] (174) at (-9, 2) {\textcolor{blue}{\textsf{2}}};
		\node [style=box] (175) at (-9, -1.75) {$V$};
		\node [style=2 control] (176) at (-9, 0.125) {\textcolor{blue}{\textsf{2}}};
		\node [style=none] (179) at (-10.25, 2) {};
		\node [style=none] (180) at (-10.25, 0.125) {};
		\node [style=none] (181) at (-10.25, -1.75) {};
		\node [style=1 control] (182) at (-7.25, 0.125) {\textcolor{red}{\textsf{1}}};
		\node [style=none] (183) at (-4.75, 0.075) {$=$};
		\node [style=none] (184) at (-3.5, 2) {};
		\node [style=none] (185) at (-3.5, 0.125) {};
		\node [style=none] (186) at (-3.5, -1.75) {};
		\node [style=box] (187) at (-2, -1.75) {$V$};
		\node [style=2 control] (188) at (-2, 2) {\textcolor{blue}{\textsf{2}}};
		\node [style=2 control] (189) at (0, 2) {\textcolor{blue}{\textsf{2}}};
		\node [style=box] (190) at (2, -1.75) {$V^{\dagger}$};
		\node [style=1 control] (191) at (2, 0.125) {\textcolor{red}{\textsf{1}}};
		\node [style=none] (192) at (7.75, 2) {};
		\node [style=none] (193) at (7.75, 0.125) {};
		\node [style=none] (194) at (7.75, -1.75) {};
		\node [style=box] (195) at (0, 0.125) {$X_{+1}$};
		\node [style=2 control] (196) at (4, 2) {\textcolor{blue}{\textsf{2}}};
		\node [style=box] (197) at (4, 0.125) {$X_{-1}$};
		\node [style=box] (198) at (6, -1.75) {$V$};
		\node [style=1 control] (199) at (6, 0.125) {\textcolor{red}{\textsf{1}}};
	\end{pgfonlayer}
	\begin{pgfonlayer}{edgelayer}
		\draw (167.center) to (171);
		\draw (169.center) to (170);
		\draw (171) to (174);
		\draw (170) to (175);
		\draw (176) to (175);
		\draw (174) to (176);
		\draw (174) to (179.center);
		\draw (180.center) to (176);
		\draw (175) to (181.center);
		\draw (182) to (176);
		\draw (168.center) to (182);
		\draw (171) to (182);
		\draw (170) to (182);
		\draw (188) to (187);
		\draw (184.center) to (188);
		\draw (186.center) to (187);
		\draw (188) to (189);
		\draw (187) to (190);
		\draw (191) to (190);
		\draw (185.center) to (195);
		\draw (189) to (195);
		\draw [in=180, out=0] (195) to (191);
		\draw (196) to (197);
		\draw (189) to (196);
		\draw (191) to (197);
		\draw (196) to (192.center);
		\draw (199) to (198);
		\draw (197) to (199);
		\draw (199) to (193.center);
		\draw (190) to (198);
		\draw (198) to (194.center);
	\end{pgfonlayer}
\end{tikzpicture}

%% file: figures/sleatorweinfurter.tikz
\begin{tikzpicture}
	\begin{pgfonlayer}{nodelayer}
		\node [style=none] (167) at (-8.5, 2) {};
		\node [style=none] (168) at (-8.5, 0.125) {};
		\node [style=none] (169) at (-8.5, -1.75) {};
		\node [style=box] (170) at (-7.25, -1.75) {$V^2$};
		\node [style=small dot] (171) at (-7.25, 2) {};
		\node [style=none] (179) at (-6, 2) {};
		\node [style=none] (180) at (-6, 0.125) {};
		\node [style=none] (181) at (-6, -1.75) {};
		\node [style=small dot] (182) at (-7.25, 0.125) {};
		\node [style=none] (183) at (-4.75, 0.075) {$=$};
		\node [style=none] (184) at (-3.5, 2) {};
		\node [style=none] (185) at (-3.5, 0.125) {};
		\node [style=none] (186) at (-3.5, -1.75) {};
		\node [style=box] (187) at (-2, -1.75) {$V$};
		\node [style=small dot] (188) at (-2, 2) {};
		\node [style=small dot] (189) at (0, 2) {};
		\node [style=box] (190) at (2, -1.75) {$V^{\dagger}$};
		\node [style=small dot] (191) at (2, 0.125) {};
		\node [style=none] (192) at (7.75, 2) {};
		\node [style=none] (193) at (7.75, 0.125) {};
		\node [style=none] (194) at (7.75, -1.75) {};
		\node [style=targ] (195) at (0, 0.125) {$+$};
		\node [style=small dot] (196) at (4, 2) {};
		\node [style=targ] (197) at (4, 0.125) {$+$};
		\node [style=box] (198) at (6, -1.75) {$V$};
		\node [style=small dot] (199) at (6, 0.125) {};
	\end{pgfonlayer}
	\begin{pgfonlayer}{edgelayer}
		\draw (167.center) to (171);
		\draw (169.center) to (170);
		\draw (168.center) to (182);
		\draw (171) to (182);
		\draw (170) to (182);
		\draw (188) to (187);
		\draw (184.center) to (188);
		\draw (186.center) to (187);
		\draw (188) to (189);
		\draw (187) to (190);
		\draw (191) to (190);
		\draw (185.center) to (195);
		\draw (189) to (195);
		\draw [in=180, out=0] (195) to (191);
		\draw (196) to (197);
		\draw (189) to (196);
		\draw (191) to (197);
		\draw (196) to (192.center);
		\draw (199) to (198);
		\draw (197) to (199);
		\draw (199) to (193.center);
		\draw (190) to (198);
		\draw (198) to (194.center);
		\draw [style=simple] (171) to (179.center);
		\draw [style=simple] (182) to (180.center);
		\draw [style=simple] (170) to (181.center);
	\end{pgfonlayer}
\end{tikzpicture}

%% file: figures/qubitgenmczphase.tikz
\begin{tikzpicture}
	\begin{pgfonlayer}{nodelayer}
		\node [style=none] (0) at (-4, 1.375) {};
		\node [style=none] (1) at (-4, -1.625) {};
		\node [style=none] (2) at (-4, -0.625) {};
		\node [style=none] (3) at (-1, 1.375) {};
		\node [style=none] (4) at (-1, -1.625) {};
		\node [style=none] (5) at (-1, -0.625) {};
		\node [style=small dot] (6) at (-2.5, -0.625) {};
		\node [style=box] (7) at (-2.5, -1.625) {$Z(\alpha)$};
		\node [style=none] (8) at (-2.5, 0.875) {};
		\node [style=none] (9) at (-2.5, -0.125) {};
		\node [style=none] (10) at (-2.5, 0.625) {$\vdots$};
		\node [style=small dot] (11) at (-2.5, 1.375) {};
		\node [style=none] (32) at (0, 0) {$\cong$};
		\node [style=none] (34) at (1, 1.375) {};
		\node [style=box] (35) at (2.75, -1.625) {$Z(\alpha/2)$};
		\node [style=none] (36) at (1, -0.625) {};
		\node [style=small dot] (37) at (4.75, -0.625) {};
		\node [style=targ] (38) at (4.75, -1.625) {$+$};
		\node [style=none] (39) at (4.75, 0.875) {};
		\node [style=none] (40) at (4.75, -0.125) {};
		\node [style=none] (41) at (4.75, 0.625) {$\vdots$};
		\node [style=small dot] (42) at (4.75, 1.375) {};
		\node [style=box] (43) at (7, -1.625) {$Z(-\alpha/2)$};
		\node [style=none] (44) at (10.5, 1.375) {};
		\node [style=none] (45) at (10.5, -1.625) {};
		\node [style=none] (46) at (10.5, -0.625) {};
		\node [style=small dot] (47) at (9.25, -0.625) {};
		\node [style=targ] (48) at (9.25, -1.625) {$+$};
		\node [style=none] (49) at (9.25, 0.875) {};
		\node [style=none] (50) at (9.25, -0.125) {};
		\node [style=none] (51) at (9.25, 0.625) {$\vdots$};
		\node [style=small dot] (52) at (9.25, 1.375) {};
		\node [style=none] (53) at (1, -1.625) {};
	\end{pgfonlayer}
	\begin{pgfonlayer}{edgelayer}
		\draw (6) to (7);
		\draw (2.center) to (6);
		\draw (6) to (5.center);
		\draw (4.center) to (7);
		\draw (7) to (1.center);
		\draw (9.center) to (6);
		\draw (0.center) to (11);
		\draw (11) to (3.center);
		\draw (8.center) to (11);
		\draw (37) to (38);
		\draw (36.center) to (37);
		\draw (38) to (35);
		\draw (40.center) to (37);
		\draw (34.center) to (42);
		\draw (39.center) to (42);
		\draw (47) to (48);
		\draw (47) to (46.center);
		\draw (45.center) to (48);
		\draw (48) to (43);
		\draw (50.center) to (47);
		\draw (52) to (44.center);
		\draw (49.center) to (52);
		\draw (43) to (38);
		\draw (53.center) to (35);
		\draw (37) to (47);
		\draw (52) to (42);
	\end{pgfonlayer}
\end{tikzpicture}

%% file: figures/two-controlled-Xplus1.tikz
\begin{tikzpicture}
	\begin{pgfonlayer}{nodelayer}
		\node [style=none] (0) at (-5.5, 1.5) {};
		\node [style=none] (1) at (-5.5, -1.5) {};
		\node [style=none] (2) at (-5.5, 0) {};
		\node [style=none] (3) at (-1.5, 1.5) {};
		\node [style=none] (4) at (-1.5, -1.5) {};
		\node [style=none] (5) at (-1.5, 0) {};
		\node [style=2 control] (6) at (-3.5, 0) {\textcolor{blue}{\textsf{2}}};
		\node [style=box] (7) at (-3.5, -1.5) {$X_{+1}$};
		\node [style=2 control] (11) at (-3.5, 1.5) {\textcolor{blue}{\textsf{2}}};
		\node [style=none] (32) at (0, 0.025) {=};
		\node [style=none] (67) at (1.25, 1.5) {};
		\node [style=box] (68) at (2.75, -1.5) {$X_{-1}$};
		\node [style=none] (69) at (1.25, 0) {};
		\node [style=2 control] (70) at (5, 0) {\textcolor{blue}{\textsf{2}}};
		\node [style=box] (71) at (5, -1.5) {$X_{12}$};
		\node [style=2 control] (75) at (5, 1.5) {\textcolor{blue}{\textsf{2}}};
		\node [style=box] (76) at (7.25, -1.5) {$X_{+1}$};
		\node [style=none] (77) at (11, 1.5) {};
		\node [style=none] (78) at (11, -1.5) {};
		\node [style=none] (79) at (11, 0) {};
		\node [style=2 control] (80) at (9.5, 0) {\textcolor{blue}{\textsf{2}}};
		\node [style=box] (81) at (9.5, -1.5) {$X_{12}$};
		\node [style=2 control] (85) at (9.5, 1.5) {\textcolor{blue}{\textsf{2}}};
		\node [style=none] (86) at (1.25, -1.5) {};
	\end{pgfonlayer}
	\begin{pgfonlayer}{edgelayer}
		\draw (6) to (7);
		\draw (2.center) to (6);
		\draw (6) to (5.center);
		\draw (4.center) to (7);
		\draw (7) to (1.center);
		\draw (0.center) to (11);
		\draw (11) to (3.center);
		\draw (70) to (71);
		\draw (69.center) to (70);
		\draw (71) to (68);
		\draw (67.center) to (75);
		\draw (80) to (81);
		\draw (80) to (79.center);
		\draw (78.center) to (81);
		\draw [in=360, out=180] (81) to (76);
		\draw (85) to (77.center);
		\draw (76) to (71);
		\draw (86.center) to (68);
		\draw (70) to (80);
		\draw (85) to (75);
		\draw (11) to (6);
		\draw (75) to (70);
		\draw (80) to (85);
	\end{pgfonlayer}
\end{tikzpicture}

%% file: figures/toffoli-decomp.tikz
\begin{tikzpicture}
	\begin{pgfonlayer}{nodelayer}
		\node [style=none] (167) at (-4.35, 1) {};
		\node [style=none] (168) at (-4.35, -0.875) {};
		\node [style=none] (169) at (-4.35, -1.75) {};
		\node [style=small dot] (171) at (-3.225, 1) {};
		\node [style=none] (179) at (-1.975, 1) {};
		\node [style=none] (180) at (-1.975, -0.875) {};
		\node [style=none] (181) at (-1.975, -1.75) {};
		\node [style=small dot] (182) at (-3.225, -0.875) {};
		\node [style=none] (183) at (0, 0.5) {$=$};
		\node [style=targ] (200) at (-3.225, -1.765) {$+$};
		\node [style=none] (202) at (-4.35, 1.875) {};
		\node [style=none] (205) at (-1.975, 1.875) {};
		\node [style=small dot] (206) at (-3.225, 1.875) {};
		\node [style=none] (207) at (-4.375, -2.5) {};
		\node [style=none] (208) at (-2, -2.5) {};
		\node [style=none] (209) at (1.775, 1) {};
		\node [style=none] (210) at (1.775, -0.875) {};
		\node [style=none] (211) at (1.775, -1.75) {};
		\node [style=none] (213) at (7.775, 1) {};
		\node [style=none] (214) at (7.775, -0.875) {};
		\node [style=none] (215) at (7.775, -1.75) {};
		\node [style=none] (219) at (1.775, 1.875) {};
		\node [style=none] (222) at (7.775, 1.875) {};
		\node [style=none] (224) at (1.75, -2.5) {};
		\node [style=none] (225) at (7.75, -2.5) {};
		\node [style=targ] (228) at (3.025, -2.515) {$+$};
		\node [style=small dot] (230) at (3.025, 1.875) {};
		\node [style=targ] (232) at (4.275, -1.765) {$+$};
		\node [style=small dot] (233) at (4.275, -2.5) {};
		\node [style=small dot] (234) at (4.275, -0.875) {};
		\node [style=none] (235) at (-4.35, 3.75) {};
		\node [style=small dot] (236) at (-3.225, 3.75) {};
		\node [style=none] (237) at (-1.975, 3.75) {};
		\node [style=none] (238) at (1.775, 3.75) {};
		\node [style=none] (239) at (7.775, 3.75) {};
		\node [style=small dot] (240) at (3.025, 3.75) {};
		\node [style=small dot] (241) at (4.275, 0.975) {};
		\node [style=targ] (242) at (5.525, -2.515) {$+$};
		\node [style=small dot] (244) at (5.525, 1.875) {};
		\node [style=small dot] (245) at (5.525, 3.75) {};
		\node [style=targ] (246) at (6.525, -1.765) {$+$};
		\node [style=small dot] (247) at (6.525, -2.5) {};
		\node [style=small dot] (248) at (6.525, -0.875) {};
		\node [style=small dot] (249) at (6.525, 0.975) {};
		\node [style=none] (329) at (3.025, 3) {$\vdots$};
		\node [style=none] (330) at (3.025, 3.3) {};
		\node [style=none] (331) at (3.025, 2.3) {};
		\node [style=none] (332) at (5.525, 3) {$\vdots$};
		\node [style=none] (333) at (5.525, 3.3) {};
		\node [style=none] (334) at (5.525, 2.3) {};
		\node [style=none] (335) at (-3.225, 3) {$\vdots$};
		\node [style=none] (336) at (-3.225, 3.3) {};
		\node [style=none] (337) at (-3.225, 2.3) {};
		\node [style=none] (338) at (-3.225, 0.25) {$\vdots$};
		\node [style=none] (339) at (-3.225, 0.55) {};
		\node [style=none] (340) at (-3.225, -0.45) {};
		\node [style=none] (341) at (4.275, 0.25) {$\vdots$};
		\node [style=none] (342) at (4.275, 0.55) {};
		\node [style=none] (343) at (4.275, -0.45) {};
		\node [style=none] (344) at (6.525, 0.25) {$\vdots$};
		\node [style=none] (345) at (6.525, 0.55) {};
		\node [style=none] (346) at (6.525, -0.45) {};
	\end{pgfonlayer}
	\begin{pgfonlayer}{edgelayer}
		\draw (167.center) to (171);
		\draw (168.center) to (182);
		\draw [style=simple] (171) to (179.center);
		\draw [style=simple] (182) to (180.center);
		\draw (182) to (200);
		\draw (169.center) to (181.center);
		\draw (202.center) to (206);
		\draw [style=simple] (206) to (205.center);
		\draw (206) to (171);
		\draw (207.center) to (208.center);
		\draw (211.center) to (215.center);
		\draw (224.center) to (225.center);
		\draw (210.center) to (214.center);
		\draw (213.center) to (209.center);
		\draw (219.center) to (222.center);
		\draw (230) to (228);
		\draw (233) to (234);
		\draw (234) to (232);
		\draw (235.center) to (236);
		\draw [style=simple] (236) to (237.center);
		\draw (239.center) to (238.center);
		\draw (244) to (242);
		\draw (247) to (248);
		\draw (248) to (246);
		\draw [style=simple] (171) to (339.center);
		\draw [style=simple] (340.center) to (182);
		\draw [style=simple] (241) to (342.center);
		\draw [style=simple] (343.center) to (234);
		\draw [style=simple] (249) to (345.center);
		\draw [style=simple] (346.center) to (248);
		\draw [style=simple] (236) to (336.center);
		\draw [style=simple] (337.center) to (206);
		\draw [style=simple] (240) to (330.center);
		\draw [style=simple] (331.center) to (230);
		\draw [style=simple] (245) to (333.center);
		\draw [style=simple] (334.center) to (244);
	\end{pgfonlayer}
\end{tikzpicture}

%% file: figures/Xplus1-decomp.tikz
\begin{tikzpicture}
	\begin{pgfonlayer}{nodelayer}
		\node [style=none] (167) at (-5.75, 0.7) {};
		\node [style=none] (168) at (-5.75, -1.175) {};
		\node [style=none] (169) at (-5.75, -2.55) {};
		\node [style=2 control] (174) at (-4, 0.7) {\textcolor{blue}{\textsf{2}}};
		\node [style=box] (175) at (-4, -2.55) {$X_{+1}$};
		\node [style=2 control] (176) at (-4, -1.175) {\textcolor{blue}{\textsf{2}}};
		\node [style=none] (179) at (-2.25, 0.7) {};
		\node [style=none] (180) at (-2.25, -1.175) {};
		\node [style=none] (181) at (-2.25, -2.55) {};
		\node [style=none] (183) at (0, 0.025) {$=$};
		\node [style=none] (185) at (1.25, 0.7) {};
		\node [style=none] (193) at (14.25, 0.7) {};
		\node [style=none] (200) at (-5.75, 3.95) {};
		\node [style=none] (201) at (-5.75, 2.075) {};
		\node [style=2 control] (202) at (-4, 3.95) {\textcolor{blue}{\textsf{2}}};
		\node [style=2 control] (203) at (-4, 2.075) {\textcolor{blue}{\textsf{2}}};
		\node [style=none] (204) at (-2.25, 3.95) {};
		\node [style=none] (205) at (-2.25, 2.075) {};
		\node [style=none] (207) at (14.25, 0.7) {};
		\node [style=none] (208) at (14.25, -1.175) {};
		\node [style=none] (209) at (14.25, -2.55) {};
		\node [style=none] (210) at (14.25, 3.95) {};
		\node [style=none] (211) at (14.25, 2.075) {};
		\node [style=none] (212) at (1.25, 0.7) {};
		\node [style=none] (213) at (1.25, -1.175) {};
		\node [style=none] (214) at (1.25, -2.55) {};
		\node [style=none] (215) at (1.25, 3.95) {};
		\node [style=none] (216) at (1.25, 2.075) {};
		\node [style=none] (217) at (-5.75, -4.05) {};
		\node [style=none] (218) at (-2.25, -4.05) {};
		\node [style=none] (219) at (-4, 3.2) {$\vdots$};
		\node [style=none] (220) at (-4, -0.05) {$\vdots$};
		\node [style=box] (222) at (2.75, -4.05) {$X_{+1}$};
		\node [style=2 control] (224) at (2.75, 3.95) {\textcolor{blue}{\textsf{2}}};
		\node [style=2 control] (225) at (2.75, 2.075) {\textcolor{blue}{\textsf{2}}};
		\node [style=none] (226) at (1.25, -4.05) {};
		\node [style=none] (227) at (14.25, -4.05) {};
		\node [style=box] (228) at (4.5, -2.55) {$X_{+1}$};
		\node [style=2 control] (229) at (4.5, 0.7) {\textcolor{blue}{\textsf{2}}};
		\node [style=2 control] (230) at (4.5, -4.025) {\textcolor{blue}{\textsf{2}}};
		\node [style=none] (231) at (2.75, 3.2) {$\vdots$};
		\node [style=none] (232) at (4.5, -0.05) {$\vdots$};
		\node [style=2 control] (233) at (4.5, -1.175) {\textcolor{blue}{\textsf{2}}};
		\node [style=none] (248) at (-4, 3.5) {};
		\node [style=none] (249) at (-4, 2.5) {};
		\node [style=none] (250) at (-4, 0.25) {};
		\node [style=none] (251) at (-4, -0.75) {};
		\node [style=none] (252) at (2.75, 3.5) {};
		\node [style=none] (253) at (2.75, 2.5) {};
		\node [style=none] (254) at (4.5, 0.25) {};
		\node [style=none] (255) at (4.5, -0.75) {};
		\node [style=box] (256) at (7, -4.05) {$X_{+1}$};
		\node [style=2 control] (257) at (7, 3.95) {\textcolor{blue}{\textsf{2}}};
		\node [style=2 control] (258) at (7, 2.075) {\textcolor{blue}{\textsf{2}}};
		\node [style=box] (259) at (8.75, -2.55) {$X_{+1}$};
		\node [style=2 control] (260) at (8.75, 0.7) {\textcolor{blue}{\textsf{2}}};
		\node [style=2 control] (261) at (8.75, -4.025) {\textcolor{blue}{\textsf{2}}};
		\node [style=none] (262) at (7, 3.2) {$\vdots$};
		\node [style=none] (263) at (8.75, -0.05) {$\vdots$};
		\node [style=2 control] (264) at (8.75, -1.175) {\textcolor{blue}{\textsf{2}}};
		\node [style=none] (265) at (7, 3.5) {};
		\node [style=none] (266) at (7, 2.5) {};
		\node [style=none] (267) at (8.75, 0.25) {};
		\node [style=none] (268) at (8.75, -0.75) {};
		\node [style=box] (269) at (11.25, -4.05) {$X_{+1}$};
		\node [style=2 control] (270) at (11.25, 3.95) {\textcolor{blue}{\textsf{2}}};
		\node [style=2 control] (271) at (11.25, 2.075) {\textcolor{blue}{\textsf{2}}};
		\node [style=box] (272) at (13, -2.55) {$X_{+1}$};
		\node [style=2 control] (273) at (13, 0.7) {\textcolor{blue}{\textsf{2}}};
		\node [style=2 control] (274) at (13, -4.025) {\textcolor{blue}{\textsf{2}}};
		\node [style=none] (275) at (11.25, 3.2) {$\vdots$};
		\node [style=none] (276) at (13, -0.05) {$\vdots$};
		\node [style=2 control] (277) at (13, -1.175) {\textcolor{blue}{\textsf{2}}};
		\node [style=none] (278) at (11.25, 3.5) {};
		\node [style=none] (279) at (11.25, 2.5) {};
		\node [style=none] (280) at (13, 0.25) {};
		\node [style=none] (281) at (13, -0.75) {};
	\end{pgfonlayer}
	\begin{pgfonlayer}{edgelayer}
		\draw (176) to (175);
		\draw (174) to (179.center);
		\draw (180.center) to (176);
		\draw (175) to (181.center);
		\draw [style=simple] (174) to (167.center);
		\draw [style=simple] (168.center) to (176);
		\draw [style=simple] (169.center) to (175);
		\draw (202) to (204.center);
		\draw (205.center) to (203);
		\draw [style=simple] (202) to (200.center);
		\draw [style=simple] (201.center) to (203);
		\draw [style=simple] (203) to (174);
		\draw (217.center) to (218.center);
		\draw (226.center) to (227.center);
		\draw (214.center) to (209.center);
		\draw (208.center) to (213.center);
		\draw (212.center) to (207.center);
		\draw (211.center) to (216.center);
		\draw (215.center) to (210.center);
		\draw (225) to (222);
		\draw (230) to (228);
		\draw (202) to (248.center);
		\draw (249.center) to (203);
		\draw (174) to (250.center);
		\draw (251.center) to (176);
		\draw (224) to (252.center);
		\draw (253.center) to (225);
		\draw (229) to (254.center);
		\draw (255.center) to (233);
		\draw (233) to (228);
		\draw (258) to (256);
		\draw (261) to (259);
		\draw (257) to (265.center);
		\draw (266.center) to (258);
		\draw (260) to (267.center);
		\draw (268.center) to (264);
		\draw (264) to (259);
		\draw (271) to (269);
		\draw (274) to (272);
		\draw (270) to (278.center);
		\draw (279.center) to (271);
		\draw (273) to (280.center);
		\draw (281.center) to (277);
		\draw (277) to (272);
	\end{pgfonlayer}
\end{tikzpicture}

%% file: figures/V-V2-many.tikz
\begin{tikzpicture}
	\begin{pgfonlayer}{nodelayer}
		\node [style=none] (167) at (-6, 1) {};
		\node [style=none] (168) at (-6, -0.375) {};
		\node [style=none] (169) at (-6, -1.75) {};
		\node [style=box] (170) at (-7.25, -1.75) {$V^2$};
		\node [style=2 control] (171) at (-7.25, 1) {\textcolor{blue}{\textsf{2}}};
		\node [style=2 control] (174) at (-9, 1) {\textcolor{blue}{\textsf{2}}};
		\node [style=box] (175) at (-9, -1.75) {$V$};
		\node [style=2 control] (176) at (-9, -0.375) {\textcolor{blue}{\textsf{2}}};
		\node [style=none] (179) at (-10.25, 1) {};
		\node [style=none] (180) at (-10.25, -0.375) {};
		\node [style=none] (181) at (-10.25, -1.75) {};
		\node [style=1 control] (182) at (-7.25, -0.375) {\textcolor{red}{\textsf{1}}};
		\node [style=none] (183) at (-4.75, -0.425) {$=$};
		\node [style=none] (184) at (-3.5, 1) {};
		\node [style=none] (185) at (-3.5, -0.375) {};
		\node [style=none] (186) at (-3.5, -1.75) {};
		\node [style=box] (187) at (-2, -1.75) {$V$};
		\node [style=box] (190) at (2, -1.75) {$V^{\dagger}$};
		\node [style=1 control] (191) at (2, -0.375) {\textcolor{red}{\textsf{1}}};
		\node [style=none] (192) at (7.75, 1) {};
		\node [style=none] (193) at (7.75, -0.375) {};
		\node [style=none] (194) at (7.75, -1.75) {};
		\node [style=box] (195) at (0, -0.375) {$X_{+1}$};
		\node [style=box] (197) at (4, -0.375) {$X_{-1}$};
		\node [style=box] (198) at (6, -1.75) {$V$};
		\node [style=1 control] (199) at (6, -0.375) {\textcolor{red}{\textsf{1}}};
		\node [style=none] (200) at (-6, 3) {};
		\node [style=2 control] (201) at (-7.25, 3) {\textcolor{blue}{\textsf{2}}};
		\node [style=2 control] (202) at (-9, 3) {\textcolor{blue}{\textsf{2}}};
		\node [style=none] (203) at (-10.25, 3) {};
		\node [style=none] (204) at (-3.5, 3) {};
		\node [style=none] (207) at (7.75, 3) {};
		\node [style=none] (209) at (-9, 2.25) {$\vdots$};
		\node [style=none] (210) at (-9, 2.5) {};
		\node [style=none] (211) at (-9, 1.5) {};
		\node [style=none] (212) at (-7.25, 2.25) {$\vdots$};
		\node [style=none] (213) at (-7.25, 2.5) {};
		\node [style=none] (214) at (-7.25, 1.5) {};
		\node [style=2 control] (215) at (-2, 1) {\textcolor{blue}{\textsf{2}}};
		\node [style=2 control] (216) at (-2, 3) {\textcolor{blue}{\textsf{2}}};
		\node [style=none] (217) at (-2, 2.25) {$\vdots$};
		\node [style=none] (218) at (-2, 2.5) {};
		\node [style=none] (219) at (-2, 1.5) {};
		\node [style=2 control] (220) at (0, 1) {\textcolor{blue}{\textsf{2}}};
		\node [style=2 control] (221) at (0, 3) {\textcolor{blue}{\textsf{2}}};
		\node [style=none] (222) at (0, 2.25) {$\vdots$};
		\node [style=none] (223) at (0, 2.5) {};
		\node [style=none] (224) at (0, 1.5) {};
		\node [style=2 control] (225) at (4, 1) {\textcolor{blue}{\textsf{2}}};
		\node [style=2 control] (226) at (4, 3) {\textcolor{blue}{\textsf{2}}};
		\node [style=none] (227) at (4, 2.25) {$\vdots$};
		\node [style=none] (228) at (4, 2.5) {};
		\node [style=none] (229) at (4, 1.5) {};
	\end{pgfonlayer}
	\begin{pgfonlayer}{edgelayer}
		\draw (167.center) to (171);
		\draw (169.center) to (170);
		\draw (171) to (174);
		\draw (170) to (175);
		\draw (176) to (175);
		\draw (174) to (176);
		\draw (174) to (179.center);
		\draw (180.center) to (176);
		\draw (175) to (181.center);
		\draw (182) to (176);
		\draw (168.center) to (182);
		\draw (171) to (182);
		\draw (170) to (182);
		\draw (186.center) to (187);
		\draw (187) to (190);
		\draw (191) to (190);
		\draw (185.center) to (195);
		\draw [in=180, out=0] (195) to (191);
		\draw (191) to (197);
		\draw (199) to (198);
		\draw (197) to (199);
		\draw (199) to (193.center);
		\draw (190) to (198);
		\draw (198) to (194.center);
		\draw (200.center) to (201);
		\draw (201) to (202);
		\draw (202) to (203.center);
		\draw (202) to (210.center);
		\draw (211.center) to (174);
		\draw (214.center) to (171);
		\draw (201) to (213.center);
		\draw (216) to (218.center);
		\draw (219.center) to (215);
		\draw (215) to (187);
		\draw (184.center) to (215);
		\draw (204.center) to (216);
		\draw (221) to (223.center);
		\draw (224.center) to (220);
		\draw (220) to (195);
		\draw (215) to (220);
		\draw (216) to (221);
		\draw (226) to (228.center);
		\draw (229.center) to (225);
		\draw (225) to (197);
		\draw (225) to (192.center);
		\draw (226) to (207.center);
		\draw (220) to (225);
		\draw (221) to (226);
	\end{pgfonlayer}
\end{tikzpicture}

%% file: figures/gentwoxp1.tikz
\begin{tikzpicture}
	\begin{pgfonlayer}{nodelayer}
		\node [style=none] (0) at (-4, 1.375) {};
		\node [style=none] (1) at (-4, -1.625) {};
		\node [style=none] (2) at (-4, -0.625) {};
		\node [style=none] (3) at (-1, 1.375) {};
		\node [style=none] (4) at (-1, -1.625) {};
		\node [style=none] (5) at (-1, -0.625) {};
		\node [style=2 control] (6) at (-2.5, -0.625) {\textcolor{blue}{\textsf{2}}};
		\node [style=box] (7) at (-2.5, -1.625) {$X_{+1}$};
		\node [style=none] (8) at (-2.5, 0.875) {};
		\node [style=none] (9) at (-2.5, -0.125) {};
		\node [style=none] (10) at (-2.5, 0.625) {$\vdots$};
		\node [style=2 control] (11) at (-2.5, 1.375) {\textcolor{blue}{\textsf{2}}};
		\node [style=none] (32) at (0, 0) {=};
		\node [style=none] (67) at (1, 1.375) {};
		\node [style=box] (68) at (2.5, -1.625) {$X_{-1}$};
		\node [style=none] (69) at (1, -0.625) {};
		\node [style=2 control] (70) at (4.75, -0.625) {\textcolor{blue}{\textsf{2}}};
		\node [style=box] (71) at (4.75, -1.625) {$X_{12}$};
		\node [style=none] (72) at (4.75, 0.875) {};
		\node [style=none] (73) at (4.75, -0.125) {};
		\node [style=none] (74) at (4.75, 0.625) {$\vdots$};
		\node [style=2 control] (75) at (4.75, 1.375) {\textcolor{blue}{\textsf{2}}};
		\node [style=box] (76) at (7, -1.625) {$X_{+1}$};
		\node [style=none] (77) at (10.5, 1.375) {};
		\node [style=none] (78) at (10.5, -1.625) {};
		\node [style=none] (79) at (10.5, -0.625) {};
		\node [style=2 control] (80) at (9.25, -0.625) {\textcolor{blue}{\textsf{2}}};
		\node [style=box] (81) at (9.25, -1.625) {$X_{12}$};
		\node [style=none] (82) at (9.25, 0.875) {};
		\node [style=none] (83) at (9.25, -0.125) {};
		\node [style=none] (84) at (9.25, 0.625) {$\vdots$};
		\node [style=2 control] (85) at (9.25, 1.375) {\textcolor{blue}{\textsf{2}}};
		\node [style=none] (86) at (1, -1.625) {};
	\end{pgfonlayer}
	\begin{pgfonlayer}{edgelayer}
		\draw (6) to (7);
		\draw (2.center) to (6);
		\draw (6) to (5.center);
		\draw (4.center) to (7);
		\draw (7) to (1.center);
		\draw (9.center) to (6);
		\draw (0.center) to (11);
		\draw (11) to (3.center);
		\draw (8.center) to (11);
		\draw (70) to (71);
		\draw (69.center) to (70);
		\draw (71) to (68);
		\draw (73.center) to (70);
		\draw (67.center) to (75);
		\draw (72.center) to (75);
		\draw (80) to (81);
		\draw (80) to (79.center);
		\draw (78.center) to (81);
		\draw [in=360, out=180] (81) to (76);
		\draw (83.center) to (80);
		\draw (85) to (77.center);
		\draw (82.center) to (85);
		\draw (76) to (71);
		\draw (86.center) to (68);
		\draw (70) to (80);
		\draw (85) to (75);
	\end{pgfonlayer}
\end{tikzpicture}

%% file: figures/gentwocx.tikz
\begin{tikzpicture}
	\begin{pgfonlayer}{nodelayer}
		\node [style=none] (3) at (2, 3.25) {};
		\node [style=none] (4) at (2, -1.5) {};
		\node [style=none] (5) at (2, 0) {};
		\node [style=2 control] (6) at (3.5, 0) {\textcolor{blue}{\textsf{2}}};
		\node [style=2 control] (7) at (3.5, -1.5) {\textcolor{blue}{\textsf{2}}};
		\node [style=none] (8) at (3.5, 2.5) {};
		\node [style=none] (9) at (3.5, 0.75) {};
		\node [style=none] (10) at (3.5, 1.75) {$\vdots$};
		\node [style=2 control] (11) at (3.5, 3.25) {\textcolor{blue}{\textsf{2}}};
		\node [style=none] (13) at (2, -3) {};
		\node [style=box] (14) at (3.5, -3) {$X_{-1}$};
		\node [style=none] (15) at (8, 3.25) {};
		\node [style=none] (16) at (8, -1.5) {};
		\node [style=none] (17) at (8, 0) {};
		\node [style=2 control] (18) at (6.5, 0) {\textcolor{blue}{\textsf{2}}};
		\node [style=1 control] (19) at (6.5, -1.5) {\textcolor{red}{\textsf{1}}};
		\node [style=none] (20) at (6.5, 2.5) {};
		\node [style=none] (21) at (6.5, 0.75) {};
		\node [style=none] (22) at (6.5, 1.75) {$\vdots$};
		\node [style=2 control] (23) at (6.5, 3.25) {\textcolor{blue}{\textsf{2}}};
		\node [style=none] (24) at (8, -3) {};
		\node [style=box] (25) at (6.5, -3) {$X_{+1}$};
		\node [style=none] (26) at (-5.5, 3.25) {};
		\node [style=none] (27) at (-5.5, -1.5) {};
		\node [style=none] (28) at (-5.5, 0) {};
		\node [style=2 control] (29) at (-3.75, 0) {\textcolor{blue}{\textsf{2}}};
		\node [style=none] (31) at (-3.75, 2.5) {};
		\node [style=none] (32) at (-3.75, 0.75) {};
		\node [style=none] (33) at (-3.75, 1.75) {$\vdots$};
		\node [style=2 control] (34) at (-3.75, 3.25) {\textcolor{blue}{\textsf{2}}};
		\node [style=none] (35) at (-5.5, -3) {};
		\node [style=none] (37) at (-2, 3.25) {};
		\node [style=none] (38) at (-2, 0) {};
		\node [style=none] (39) at (-2, -1.5) {};
		\node [style=none] (40) at (-2, -3) {};
		\node [style=none] (41) at (-5, -0.875) {};
		\node [style=none] (42) at (-2.5, -0.875) {};
		\node [style=none] (43) at (-5, -3.875) {};
		\node [style=none] (44) at (-2.5, -3.875) {};
		\node [style=none] (45) at (-3.75, -0.875) {};
		\node [style=none] (46) at (0, 0) {$=$};
		\node [style=vertex set] (47) at (-3.75, -1.5) {$\Lambda$};
		\node [style=box] (48) at (-3.75, -3) {$X_{+1}$};
	\end{pgfonlayer}
	\begin{pgfonlayer}{edgelayer}
		\draw (6) to (7);
		\draw (6) to (5.center);
		\draw (4.center) to (7);
		\draw (9.center) to (6);
		\draw (11) to (3.center);
		\draw (8.center) to (11);
		\draw (13.center) to (14);
		\draw (7) to (14);
		\draw (18) to (19);
		\draw (17.center) to (18);
		\draw (19) to (16.center);
		\draw (21.center) to (18);
		\draw (15.center) to (23);
		\draw (20.center) to (23);
		\draw (25) to (24.center);
		\draw (19) to (25);
		\draw (18) to (6);
		\draw (7) to (19);
		\draw (23) to (11);
		\draw (14) to (25);
		\draw (28.center) to (29);
		\draw (32.center) to (29);
		\draw (26.center) to (34);
		\draw (31.center) to (34);
		\draw (29) to (38.center);
		\draw (34) to (37.center);
		\draw [style=dashed edge] (45.center)
			 to (41.center)
			 to (43.center)
			 to (44.center)
			 to (42.center)
			 to cycle;
		\draw (29) to (45.center);
		\draw [style=simple] (47) to (48);
		\draw [style=simple] (27.center) to (47);
		\draw [style=simple] (47) to (39.center);
		\draw [style=simple] (40.center) to (48);
		\draw [style=simple] (48) to (35.center);
	\end{pgfonlayer}
\end{tikzpicture}

%% file: figures/gentwot.tikz
\begin{tikzpicture}
	\begin{pgfonlayer}{nodelayer}
		\node [style=none] (1) at (-5.5, 2.35) {};
		\node [style=none] (3) at (-5.5, -2.4) {};
		\node [style=none] (4) at (-5.5, -0.9) {};
		\node [style=none] (6) at (-1.5, 2.35) {};
		\node [style=none] (8) at (-1.5, -2.4) {};
		\node [style=none] (9) at (-1.5, -0.9) {};
		\node [style=2 control] (13) at (-3.5, -0.9) {\textcolor{blue}{\textsf{2}}};
		\node [style=box] (14) at (-3.5, -2.4) {$T$};
		\node [style=none] (21) at (-3.5, 1.6) {};
		\node [style=none] (22) at (-3.5, -0.15) {};
		\node [style=none] (23) at (-3.5, 0.85) {$\vdots$};
		\node [style=2 control] (24) at (-3.5, 2.35) {\textcolor{blue}{\textsf{2}}};
		\node [style=none] (25) at (1.5, 2.35) {};
		\node [style=box] (26) at (3, -2.4) {$T^5$};
		\node [style=none] (27) at (1.5, -0.9) {};
		\node [style=2 control] (31) at (5, -0.9) {\textcolor{blue}{\textsf{2}}};
		\node [style=box] (32) at (5, -2.4) {$X_{12}$};
		\node [style=none] (33) at (5, 1.6) {};
		\node [style=none] (34) at (5, -0.15) {};
		\node [style=none] (35) at (5, 0.85) {$\vdots$};
		\node [style=2 control] (36) at (5, 2.35) {\textcolor{blue}{\textsf{2}}};
		\node [style=box] (38) at (7, -2.4) {$T^4$};
		\node [style=none] (40) at (10.5, 2.35) {};
		\node [style=none] (41) at (10.5, -2.4) {};
		\node [style=none] (42) at (10.5, -0.9) {};
		\node [style=2 control] (43) at (9, -0.9) {\textcolor{blue}{\textsf{2}}};
		\node [style=box] (44) at (9, -2.4) {$X_{12}$};
		\node [style=none] (45) at (9, 1.6) {};
		\node [style=none] (46) at (9, -0.15) {};
		\node [style=none] (47) at (9, 0.85) {$\vdots$};
		\node [style=2 control] (48) at (9, 2.35) {\textcolor{blue}{\textsf{2}}};
		\node [style=none] (49) at (1.5, -2.4) {};
		\node [style=none] (50) at (0, -0.025) {=};
	\end{pgfonlayer}
	\begin{pgfonlayer}{edgelayer}
		\draw (13) to (14);
		\draw (4.center) to (13);
		\draw (13) to (9.center);
		\draw (8.center) to (14);
		\draw (14) to (3.center);
		\draw (22.center) to (13);
		\draw (1.center) to (24);
		\draw (24) to (6.center);
		\draw (21.center) to (24);
		\draw (31) to (32);
		\draw (27.center) to (31);
		\draw [in=360, out=180] (32) to (26);
		\draw (34.center) to (31);
		\draw (25.center) to (36);
		\draw (33.center) to (36);
		\draw (43) to (44);
		\draw (43) to (42.center);
		\draw (41.center) to (44);
		\draw (44) to (38);
		\draw (46.center) to (43);
		\draw (48) to (40.center);
		\draw (45.center) to (48);
		\draw (38) to (32);
		\draw (49.center) to (26);
		\draw (31) to (43);
		\draw (48) to (36);
	\end{pgfonlayer}
\end{tikzpicture}

%% file: figures/gentwosrel.tikz
\begin{tikzpicture}
	\begin{pgfonlayer}{nodelayer}
		\node [style=none] (0) at (-5.5, 2.375) {};
		\node [style=none] (1) at (-5.5, -2.375) {};
		\node [style=none] (2) at (-5.5, -0.875) {};
		\node [style=none] (3) at (-1.5, 2.375) {};
		\node [style=none] (4) at (-1.5, -2.375) {};
		\node [style=none] (5) at (-1.5, -0.875) {};
		\node [style=2 control] (6) at (-3.5, -0.875) {\textcolor{blue}{\textsf{2}}};
		\node [style=box] (7) at (-3.5, -2.375) {$S$};
		\node [style=none] (8) at (-3.5, 1.625) {};
		\node [style=none] (9) at (-3.5, -0.125) {};
		\node [style=none] (10) at (-3.5, 0.875) {$\vdots$};
		\node [style=2 control] (11) at (-3.5, 2.375) {\textcolor{blue}{\textsf{2}}};
		\node [style=none] (13) at (0, 0) {=};
		\node [style=none] (46) at (1.5, 2.375) {};
		\node [style=box] (47) at (3, -2.375) {$X_{01}$};
		\node [style=none] (48) at (1.5, -0.875) {};
		\node [style=2 control] (49) at (9, -0.875) {\textcolor{blue}{\textsf{2}}};
		\node [style=box] (50) at (9, -2.375) {$X_{+1}$};
		\node [style=none] (51) at (9, 1.625) {};
		\node [style=none] (52) at (9, -0.125) {};
		\node [style=none] (53) at (9, 0.875) {$\vdots$};
		\node [style=2 control] (54) at (9, 2.375) {\textcolor{blue}{\textsf{2}}};
		\node [style=none] (65) at (1.5, -2.375) {};
		\node [style=box] (66) at (5, -2.375) {$T^{\dagger}$};
		\node [style=box] (67) at (7, -2.375) {$X_{01}$};
		\node [style=box] (68) at (11.25, -2.375) {$X_{01}$};
		\node [style=2 control] (69) at (17.25, -0.875) {\textcolor{blue}{\textsf{2}}};
		\node [style=box] (70) at (17.25, -2.375) {$X_{-1}$};
		\node [style=none] (71) at (17.25, 1.625) {};
		\node [style=none] (72) at (17.25, -0.125) {};
		\node [style=none] (73) at (17.25, 0.875) {$\vdots$};
		\node [style=2 control] (74) at (17.25, 2.375) {\textcolor{blue}{\textsf{2}}};
		\node [style=box] (75) at (13.25, -2.375) {$T$};
		\node [style=box] (76) at (15.25, -2.375) {$X_{01}$};
		\node [style=none] (79) at (18.75, -2.375) {};
		\node [style=none] (198) at (18.75, 2.375) {};
		\node [style=none] (199) at (18.75, -0.875) {};
	\end{pgfonlayer}
	\begin{pgfonlayer}{edgelayer}
		\draw (6) to (7);
		\draw (2.center) to (6);
		\draw (6) to (5.center);
		\draw (4.center) to (7);
		\draw (7) to (1.center);
		\draw (9.center) to (6);
		\draw (0.center) to (11);
		\draw (11) to (3.center);
		\draw (8.center) to (11);
		\draw (49) to (50);
		\draw (48.center) to (49);
		\draw (52.center) to (49);
		\draw (46.center) to (54);
		\draw (51.center) to (54);
		\draw (65.center) to (47);
		\draw (69) to (70);
		\draw (72.center) to (69);
		\draw (71.center) to (74);
		\draw (50) to (68);
		\draw (47) to (66);
		\draw (66) to (67);
		\draw (67) to (50);
		\draw (68) to (75);
		\draw (75) to (76);
		\draw (76) to (70);
		\draw (49) to (69);
		\draw (54) to (74);
		\draw [style=simple] (74) to (198.center);
		\draw [style=simple] (69) to (199.center);
		\draw [style=simple] (70) to (79.center);
	\end{pgfonlayer}
\end{tikzpicture}

%% file: figures/tcz22phase.tikz
\begin{tikzpicture}
	\begin{pgfonlayer}{nodelayer}
		\node [style=none] (1) at (-5.5, -2.375) {};
		\node [style=none] (2) at (-5.5, -0.875) {};
		\node [style=none] (4) at (-1.5, -2.375) {};
		\node [style=none] (5) at (-1.5, -0.875) {};
		\node [style=2 control] (6) at (-3.5, -0.875) {\textcolor{blue}{\textsf{2}}};
		\node [style=box] (7) at (-3.5, -2.375) {$\zeta^2 Z(2,2)$};
		\node [style=none] (13) at (0, 0) {$=$};
		\node [style=box] (47) at (7.5, -2.375) {$X_{02}$};
		\node [style=none] (48) at (9.5, -0.875) {};
		\node [style=2 control] (49) at (5.5, -0.875) {\textcolor{blue}{\textsf{2}}};
		\node [style=box] (50) at (5.5, -2.375) {$S$};
		\node [style=none] (65) at (9.5, -2.375) {};
		\node [style=box] (76) at (3.5, -2.375) {$X_{02}$};
		\node [style=none] (79) at (1.5, -2.375) {};
		\node [style=none] (199) at (1.5, -0.875) {};
		\node [style=none] (231) at (-5.5, 2.375) {};
		\node [style=none] (232) at (-1.5, 2.375) {};
		\node [style=none] (233) at (-3.5, 1.625) {};
		\node [style=none] (234) at (-3.5, -0.125) {};
		\node [style=none] (235) at (-3.5, 0.875) {$\vdots$};
		\node [style=2 control] (236) at (-3.5, 2.375) {\textcolor{blue}{\textsf{2}}};
		\node [style=none] (237) at (5.5, 1.625) {};
		\node [style=none] (238) at (5.5, -0.125) {};
		\node [style=none] (239) at (5.5, 0.875) {$\vdots$};
		\node [style=2 control] (240) at (5.5, 2.375) {\textcolor{blue}{\textsf{2}}};
		\node [style=none] (241) at (1.5, 2.375) {};
		\node [style=none] (244) at (9.5, 2.375) {};
	\end{pgfonlayer}
	\begin{pgfonlayer}{edgelayer}
		\draw (6) to (7);
		\draw (2.center) to (6);
		\draw (6) to (5.center);
		\draw (4.center) to (7);
		\draw (7) to (1.center);
		\draw (49) to (50);
		\draw (48.center) to (49);
		\draw (65.center) to (47);
		\draw [style=simple] (199.center) to (49);
		\draw [style=simple] (47) to (50);
		\draw [style=simple] (50) to (76);
		\draw [style=simple] (76) to (79.center);
		\draw (231.center) to (236);
		\draw (236) to (232.center);
		\draw (233.center) to (236);
		\draw (237.center) to (240);
		\draw [style=simple] (241.center) to (240);
		\draw [style=simple] (234.center) to (6);
		\draw [style=simple] (238.center) to (49);
		\draw [style=simple] (244.center) to (240);
	\end{pgfonlayer}
\end{tikzpicture}

%% file: figures/gentwosborrow.tikz
\begin{tikzpicture}
	\begin{pgfonlayer}{nodelayer}
		\node [style=none] (0) at (-6, 2.25) {};
		\node [style=none] (1) at (-6, -2.5) {};
		\node [style=none] (2) at (-6, -1) {};
		\node [style=none] (3) at (-2, 2.25) {};
		\node [style=none] (4) at (-2, -2.5) {};
		\node [style=none] (5) at (-2, -1) {};
		\node [style=2 control] (6) at (-4, -1) {\textcolor{blue}{\textsf{2}}};
		\node [style=box] (7) at (-4, -2.5) {$Z(0,1)$};
		\node [style=none] (8) at (-4, 1.5) {};
		\node [style=none] (9) at (-4, -0.25) {};
		\node [style=none] (10) at (-4, 0.75) {$\vdots$};
		\node [style=2 control] (11) at (-4, 2.25) {\textcolor{blue}{\textsf{2}}};
		\node [style=none] (12) at (-0.5, -1) {=};
		\node [style=none] (79) at (12.75, -4) {};
		\node [style=box] (169) at (3, -4) {$X_{02}$};
		\node [style=2 control] (174) at (5.5, -1) {\textcolor{blue}{\textsf{2}}};
		\node [style=box] (175) at (5.5, -4) {$Z_{-1}$};
		\node [style=none] (176) at (5.5, 1.5) {};
		\node [style=none] (177) at (5.5, -0.25) {};
		\node [style=none] (178) at (5.5, 0.75) {$\vdots$};
		\node [style=2 control] (179) at (5.5, 2.25) {\textcolor{blue}{\textsf{2}}};
		\node [style=box] (181) at (8, -4) {$X_{02}$};
		\node [style=2 control] (186) at (10.5, -1) {\textcolor{blue}{\textsf{2}}};
		\node [style=box] (187) at (10.5, -4) {$Z_{-1}$};
		\node [style=none] (188) at (10.5, 1.5) {};
		\node [style=none] (189) at (10.5, -0.25) {};
		\node [style=none] (190) at (10.5, 0.75) {$\vdots$};
		\node [style=2 control] (191) at (10.5, 2.25) {\textcolor{blue}{\textsf{2}}};
		\node [style=none] (198) at (12.75, 2.25) {};
		\node [style=none] (199) at (12.75, -1) {};
		\node [style=none] (205) at (0.75, -4) {};
		\node [style=none] (206) at (0.75, -1) {};
		\node [style=none] (207) at (0.75, 2.25) {};
		\node [style=none] (208) at (-6, -4) {};
		\node [style=none] (209) at (-2, -4) {};
		\node [style=2 control] (211) at (5.5, -2.5) {\textcolor{blue}{\textsf{2}}};
		\node [style=2 control] (213) at (10.5, -2.5) {\textcolor{blue}{\textsf{2}}};
		\node [style=none] (214) at (12.75, -2.5) {};
		\node [style=none] (215) at (0.75, -2.5) {};
	\end{pgfonlayer}
	\begin{pgfonlayer}{edgelayer}
		\draw (6) to (7);
		\draw (2.center) to (6);
		\draw (6) to (5.center);
		\draw (4.center) to (7);
		\draw (7) to (1.center);
		\draw (9.center) to (6);
		\draw (0.center) to (11);
		\draw (11) to (3.center);
		\draw (8.center) to (11);
		\draw (177.center) to (174);
		\draw (176.center) to (179);
		\draw (189.center) to (186);
		\draw (188.center) to (191);
		\draw [style=simple] (191) to (198.center);
		\draw [style=simple] (186) to (199.center);
		\draw [style=simple] (205.center) to (169);
		\draw [style=simple] (208.center) to (209.center);
		\draw [style=simple] (213) to (214.center);
		\draw [style=simple] (174) to (211);
		\draw [style=simple] (211) to (175);
		\draw [style=simple] (186) to (213);
		\draw [style=simple] (213) to (187);
		\draw [style=simple] (169) to (175);
		\draw [style=simple] (175) to (181);
		\draw [style=simple] (187) to (181);
		\draw [style=simple] (187) to (79.center);
		\draw [style=simple] (207.center) to (179);
		\draw [style=simple] (179) to (191);
		\draw [style=simple] (206.center) to (174);
		\draw [style=simple] (174) to (186);
		\draw [style=simple] (215.center) to (211);
		\draw [style=simple] (211) to (213);
	\end{pgfonlayer}
\end{tikzpicture}

%% file: figures/gentwosh.tikz
\begin{tikzpicture}
	\begin{pgfonlayer}{nodelayer}
		\node [style=none] (1) at (-5.25, -2.875) {};
		\node [style=none] (2) at (-5.25, -1.375) {};
		\node [style=none] (4) at (-1.25, -2.875) {};
		\node [style=none] (5) at (-1.25, -1.375) {};
		\node [style=2 control] (6) at (-3.25, -1.375) {\textcolor{blue}{\textsf{2}}};
		\node [style=box] (7) at (-3.25, -2.875) {$H$};
		\node [style=none] (13) at (0, 0) {$=$};
		\node [style=none] (48) at (1.25, -1.375) {};
		\node [style=2 control] (49) at (7, -1.375) {\textcolor{blue}{\textsf{2}}};
		\node [style=box] (50) at (7, -2.875) {$\zeta^2 Z(2,2)$};
		\node [style=none] (65) at (1.25, -2.875) {};
		\node [style=none] (79) at (23.25, -2.875) {};
		\node [style=none] (199) at (23.25, -1.375) {};
		\node [style=2 control] (202) at (13.75, -1.375) {\textcolor{blue}{\textsf{2}}};
		\node [style=box] (203) at (13.75, -2.875) {$\zeta^2 X(2,2)$};
		\node [style=2 control] (206) at (20.25, -1.375) {\textcolor{blue}{\textsf{2}}};
		\node [style=box] (207) at (20.25, -2.875) {$\zeta^2 Z(2,2)$};
		\node [style=box] (210) at (3, -1.375) {$Z(0,1)$};
		\node [style=none] (212) at (-5.25, 3.375) {};
		\node [style=none] (214) at (-1.25, 3.375) {};
		\node [style=none] (217) at (-3.25, 2.625) {};
		\node [style=none] (219) at (-3.25, 1.875) {$\vdots$};
		\node [style=2 control] (220) at (-3.25, 3.375) {\textcolor{blue}{\textsf{2}}};
		\node [style=none] (221) at (1.25, 3.375) {};
		\node [style=none] (223) at (7, 2.625) {};
		\node [style=none] (225) at (7, 1.875) {$\vdots$};
		\node [style=2 control] (226) at (7, 3.375) {\textcolor{blue}{\textsf{2}}};
		\node [style=none] (229) at (13.75, 2.625) {};
		\node [style=none] (231) at (13.75, 1.875) {$\vdots$};
		\node [style=2 control] (232) at (13.75, 3.375) {\textcolor{blue}{\textsf{2}}};
		\node [style=none] (234) at (23.25, 3.375) {};
		\node [style=none] (235) at (20.25, 2.625) {};
		\node [style=none] (237) at (20.25, 1.875) {$\vdots$};
		\node [style=2 control] (238) at (20.25, 3.375) {\textcolor{blue}{\textsf{2}}};
		\node [style=none] (239) at (3, 2.625) {};
		\node [style=none] (241) at (3, 1.875) {$\vdots$};
		\node [style=2 control] (242) at (3, 3.375) {\textcolor{blue}{\textsf{2}}};
		\node [style=2 control] (243) at (7, 0.125) {\textcolor{blue}{\textsf{2}}};
		\node [style=none] (244) at (23.25, 0.125) {};
		\node [style=2 control] (245) at (13.75, 0.125) {\textcolor{blue}{\textsf{2}}};
		\node [style=2 control] (246) at (20.25, 0.125) {\textcolor{blue}{\textsf{2}}};
		\node [style=none] (247) at (7, 0.875) {};
		\node [style=none] (248) at (13.75, 0.875) {};
		\node [style=none] (249) at (20.25, 0.875) {};
		\node [style=2 control] (250) at (3, 0.125) {\textcolor{blue}{\textsf{2}}};
		\node [style=none] (254) at (3, 0.875) {};
		\node [style=none] (255) at (1.25, 0.125) {};
		\node [style=2 control] (256) at (-3.25, 0.125) {\textcolor{blue}{\textsf{2}}};
		\node [style=none] (260) at (-3.25, 0.875) {};
		\node [style=none] (261) at (-5.25, 0.125) {};
		\node [style=none] (262) at (-1.25, 0.125) {};
		\node [style=none] (269) at (0, -9.5) {$=$};
		\node [style=none] (270) at (1.25, -10.875) {};
		\node [style=2 control] (271) at (7, -10.875) {\textcolor{blue}{\textsf{2}}};
		\node [style=box] (272) at (7, -12.375) {$\zeta^2 Z(2,2)$};
		\node [style=none] (273) at (1.25, -12.375) {};
		\node [style=none] (274) at (23.25, -12.375) {};
		\node [style=none] (275) at (23.25, -10.875) {};
		\node [style=2 control] (276) at (13.75, -10.875) {\textcolor{blue}{\textsf{2}}};
		\node [style=box] (277) at (13.75, -12.375) {$\zeta^2 X(2,2)$};
		\node [style=2 control] (278) at (20.25, -10.875) {\textcolor{blue}{\textsf{2}}};
		\node [style=box] (279) at (20.25, -12.375) {$\zeta^2 Z(2,2)$};
		\node [style=2 control] (280) at (3, -10.875) {\textcolor{blue}{\textsf{2}}};
		\node [style=none] (286) at (1.25, -6.125) {};
		\node [style=none] (287) at (7, -6.875) {};
		\node [style=none] (288) at (7, -7.625) {$\vdots$};
		\node [style=2 control] (289) at (7, -6.125) {\textcolor{blue}{\textsf{2}}};
		\node [style=none] (290) at (13.75, -6.875) {};
		\node [style=none] (291) at (13.75, -7.625) {$\vdots$};
		\node [style=2 control] (292) at (13.75, -6.125) {\textcolor{blue}{\textsf{2}}};
		\node [style=none] (293) at (23.25, -6.125) {};
		\node [style=none] (294) at (20.25, -6.875) {};
		\node [style=none] (295) at (20.25, -7.625) {$\vdots$};
		\node [style=2 control] (296) at (20.25, -6.125) {\textcolor{blue}{\textsf{2}}};
		\node [style=none] (297) at (3, -6.875) {};
		\node [style=none] (298) at (3, -7.625) {$\vdots$};
		\node [style=2 control] (299) at (3, -6.125) {\textcolor{blue}{\textsf{2}}};
		\node [style=2 control] (300) at (7, -9.375) {\textcolor{blue}{\textsf{2}}};
		\node [style=none] (301) at (23.25, -9.375) {};
		\node [style=2 control] (302) at (13.75, -9.375) {\textcolor{blue}{\textsf{2}}};
		\node [style=2 control] (303) at (20.25, -9.375) {\textcolor{blue}{\textsf{2}}};
		\node [style=none] (304) at (7, -8.625) {};
		\node [style=none] (305) at (13.75, -8.625) {};
		\node [style=none] (306) at (20.25, -8.625) {};
		\node [style=2 control] (307) at (3, -9.375) {\textcolor{blue}{\textsf{2}}};
		\node [style=none] (308) at (3, -8.625) {};
		\node [style=none] (309) at (1.25, -9.375) {};
		\node [style=box] (310) at (3, -12.375) {$\zeta^3 \mathbb{I}$};
	\end{pgfonlayer}
	\begin{pgfonlayer}{edgelayer}
		\draw (6) to (7);
		\draw (2.center) to (6);
		\draw (6) to (5.center);
		\draw (4.center) to (7);
		\draw (7) to (1.center);
		\draw (49) to (50);
		\draw (48.center) to (49);
		\draw (202) to (203);
		\draw (206) to (207);
		\draw (199.center) to (49);
		\draw (212.center) to (220);
		\draw (220) to (214.center);
		\draw (217.center) to (220);
		\draw (223.center) to (226);
		\draw (229.center) to (232);
		\draw (238) to (234.center);
		\draw (235.center) to (238);
		\draw [style=simple] (226) to (232);
		\draw [style=simple] (232) to (238);
		\draw (239.center) to (242);
		\draw [style=simple] (221.center) to (242);
		\draw [style=simple] (242) to (226);
		\draw (244.center) to (243);
		\draw [style=simple] (247.center) to (243);
		\draw [style=simple] (248.center) to (245);
		\draw [style=simple] (249.center) to (246);
		\draw [style=simple] (243) to (49);
		\draw [style=simple] (245) to (202);
		\draw [style=simple] (246) to (206);
		\draw [style=simple] (254.center) to (250);
		\draw [style=simple] (250) to (210);
		\draw [style=simple] (255.center) to (250);
		\draw [style=simple] (250) to (243);
		\draw [style=simple] (260.center) to (256);
		\draw [style=simple] (256) to (6);
		\draw [style=simple] (261.center) to (256);
		\draw [style=simple] (256) to (262.center);
		\draw [style=simple] (65.center) to (50);
		\draw [style=simple] (50) to (203);
		\draw [style=simple] (203) to (207);
		\draw [style=simple] (207) to (79.center);
		\draw (271) to (272);
		\draw (270.center) to (271);
		\draw (276) to (277);
		\draw (278) to (279);
		\draw (275.center) to (271);
		\draw (287.center) to (289);
		\draw (290.center) to (292);
		\draw (296) to (293.center);
		\draw (294.center) to (296);
		\draw [style=simple] (289) to (292);
		\draw [style=simple] (292) to (296);
		\draw (297.center) to (299);
		\draw [style=simple] (286.center) to (299);
		\draw [style=simple] (299) to (289);
		\draw (301.center) to (300);
		\draw [style=simple] (304.center) to (300);
		\draw [style=simple] (305.center) to (302);
		\draw [style=simple] (306.center) to (303);
		\draw [style=simple] (300) to (271);
		\draw [style=simple] (302) to (276);
		\draw [style=simple] (303) to (278);
		\draw [style=simple] (308.center) to (307);
		\draw [style=simple] (307) to (280);
		\draw [style=simple] (309.center) to (307);
		\draw [style=simple] (307) to (300);
		\draw [style=simple] (272) to (277);
		\draw [style=simple] (277) to (279);
		\draw [style=simple] (279) to (274.center);
		\draw [style=simple] (273.center) to (310);
		\draw [style=simple] (310) to (272);
		\draw [style=simple] (280) to (310);
	\end{pgfonlayer}
\end{tikzpicture}

%% file: figures/permutetwotritstrings.tikz
\begin{tikzpicture}
	\begin{pgfonlayer}{nodelayer}
		\node [style=none] (1) at (-5.75, 2.35) {};
		\node [style=none] (3) at (-5.75, -2.65) {};
		\node [style=none] (4) at (-5.75, -0.9) {};
		\node [style=none] (6) at (-1.5, 2.35) {};
		\node [style=none] (8) at (-1.5, -2.65) {};
		\node [style=none] (9) at (-1.5, -0.9) {};
		\node [style=0 control] (31) at (8.25, -0.9) {$a_{n-1}$};
		\node [style=box] (32) at (8.25, -2.65) {$X_{a_n,b_n}$};
		\node [style=none] (33) at (8.25, 1.6) {};
		\node [style=none] (34) at (8.25, -0.15) {};
		\node [style=none] (35) at (8.25, 0.85) {$\vdots$};
		\node [style=0 control] (36) at (8.25, 2.35) {$a_1$};
		\node [style=none] (40) at (15, 2.35) {};
		\node [style=none] (41) at (15, -2.65) {};
		\node [style=none] (42) at (15, -0.9) {};
		\node [style=box] (43) at (12.5, -0.9) {$X_{a_{n-1},b_{n-1}}$};
		\node [style=0 control] (44) at (12.5, -2.65) {$b_n$};
		\node [style=none] (45) at (12.5, 1.6) {};
		\node [style=none] (46) at (12.5, -0.15) {};
		\node [style=none] (47) at (12.5, 0.85) {$\vdots$};
		\node [style=box] (48) at (12.5, 2.35) {$X_{a_1,b_1}$};
		\node [style=none] (49) at (1.5, -2.65) {};
		\node [style=none] (50) at (0, -0.025) {=};
		\node [style=none] (51) at (-4.875, 3) {};
		\node [style=none] (52) at (-2.375, 3) {};
		\node [style=none] (53) at (-4.875, -3.25) {};
		\node [style=none] (54) at (-2.375, -3.25) {};
		\node [style=none] (55) at (-4.875, 2.35) {};
		\node [style=none] (56) at (-2.375, 2.35) {};
		\node [style=none] (57) at (-4.875, -2.65) {};
		\node [style=none] (58) at (-2.375, -2.65) {};
		\node [style=none] (59) at (-4.875, -0.9) {};
		\node [style=none] (60) at (-2.375, -0.9) {};
		\node [style=none] (61) at (-3.625, 0) {$\left(\vec{a},\vec{b}\right)$};
		\node [style=none] (62) at (-5.375, 0.85) {$\vdots$};
		\node [style=none] (63) at (-1.875, 0.85) {$\vdots$};
		\node [style=none] (64) at (1.5, 2.35) {};
		\node [style=none] (65) at (1.5, -2.65) {};
		\node [style=none] (66) at (1.5, -0.9) {};
		\node [style=box] (67) at (4, -0.9) {$X_{a_{n-1},b_{n-1}}$};
		\node [style=0 control] (68) at (4, -2.65) {$b_n$};
		\node [style=none] (69) at (4, 1.6) {};
		\node [style=none] (70) at (4, -0.15) {};
		\node [style=none] (71) at (4, 0.85) {$\vdots$};
		\node [style=box] (72) at (4, 2.35) {$X_{a_1,b_1}$};
		\node [style=none] (73) at (2, -3.55) {};
		\node [style=none] (74) at (6, -3.55) {};
		\node [style=XD phase dot] (77) at (4, -4.25) {$\small1$};
		\node [style=none] (78) at (2, -3.75) {};
		\node [style=none] (79) at (6, -3.75) {};
		\node [style=none] (80) at (6.25, -3.55) {};
		\node [style=none] (81) at (10.25, -3.55) {};
		\node [style=XD phase dot] (82) at (8.25, -4.25) {$2$};
		\node [style=none] (83) at (6.25, -3.75) {};
		\node [style=none] (84) at (10.25, -3.75) {};
		\node [style=none] (85) at (10.5, -3.55) {};
		\node [style=none] (86) at (14.5, -3.55) {};
		\node [style=XD phase dot] (87) at (12.5, -4.25) {$3$};
		\node [style=none] (88) at (10.5, -3.75) {};
		\node [style=none] (89) at (14.5, -3.75) {};
	\end{pgfonlayer}
	\begin{pgfonlayer}{edgelayer}
		\draw (31) to (32);
		\draw (34.center) to (31);
		\draw (33.center) to (36);
		\draw (43) to (44);
		\draw (43) to (42.center);
		\draw (41.center) to (44);
		\draw (46.center) to (43);
		\draw (48) to (40.center);
		\draw (45.center) to (48);
		\draw [style=simple] (51.center) to (53.center);
		\draw [style=simple] (53.center) to (54.center);
		\draw [style=simple] (54.center) to (52.center);
		\draw [style=simple] (52.center) to (51.center);
		\draw [style=simple] (1.center) to (55.center);
		\draw [style=simple] (4.center) to (59.center);
		\draw [style=simple] (3.center) to (57.center);
		\draw [style=simple] (58.center) to (8.center);
		\draw [style=simple] (9.center) to (60.center);
		\draw [style=simple] (6.center) to (56.center);
		\draw (67) to (68);
		\draw [in=360, out=180] (67) to (66.center);
		\draw (65.center) to (68);
		\draw (70.center) to (67);
		\draw (72) to (64.center);
		\draw (69.center) to (72);
		\draw [style=simple] (72) to (36);
		\draw [style=simple] (36) to (48);
		\draw [style=simple, in=180, out=0] (67) to (31);
		\draw [style=simple, in=180, out=0] (31) to (43);
		\draw [style=simple] (68) to (32);
		\draw [style=simple] (32) to (44);
		\draw [style=thin] (73.center) to (78.center);
		\draw [style=thin] (78.center) to (79.center);
		\draw [style=thin] (79.center) to (74.center);
		\draw [style=thin] (80.center) to (83.center);
		\draw [style=thin] (83.center) to (84.center);
		\draw [style=thin] (84.center) to (81.center);
		\draw [style=thin] (85.center) to (88.center);
		\draw [style=thin] (88.center) to (89.center);
		\draw [style=thin] (89.center) to (86.center);
	\end{pgfonlayer}
\end{tikzpicture}